\documentclass[a4paper,onecolumn,accepted=2019-12-24]{quantumarticle}
\pdfoutput=1
\usepackage{graphicx}
\usepackage{amsmath}
\usepackage{amssymb}
\usepackage{amsthm,dsfont}
\usepackage{color}

\theoremstyle{plain}
\newtheorem{thm}{Theorem}
\newtheorem{introthm}{Theorem}
\newtheorem{conj}{Conjecture}
\newtheorem{lemma}{Lemma}
\newtheorem{defin}{Definition}

\newtheorem*{test}{Hypothesis testing scenario}

\newcommand{\defn}[1]{
\begin{defin}
#1
\end{defin}}

\newcommand{\mc}[1]{\mathcal{#1}}

\newcommand{\tr}{\text{tr}}
\newcommand{\abs}[1]{\left|#1\right|}
\newcommand{\pr}{{\rm Pr}}

\newcommand{\bitstring}[1]{\left\{0,1\right\}^{#1}}
\newcommand{\set}[1]{\left\{#1\right\}}

\newcommand{\dist}[2]{||#1 - #2||_1}

\newcommand{\bbn}{\mathbb{N}}

\newcommand{\bbr}{\mathbb{R}}

\newcommand{\ts}{$t$-sparse }
\newcommand{\ets}{$\epsilon$-approximately $t$-sparse}
\newcommand{\esim}{$\epsilon$-simulation}
\newcommand{\esimable}{$\epsilon$-simulable}
\newcommand{\esimer}{$\epsilon$-simulator}

\newcommand{\esimed}{$\epsilon$-simulated}
\newcommand{\cmany}{$\mathcal{C}_{\rm PROD}$}
\newcommand{\ciqp}{$\mathcal{C}_{\rm IQP}$}
\newcommand{\clon}{$\mathcal{C}_{\rm LON}$}

\newcommand{\cuniv}{$\mathcal{C}_{\rm UNIV}$}
\newcommand{\cpolyneg}{$\mathcal{C}_{poly\mathcal{N}}$}
\newcommand{\cencode}{$\mathcal{C}_{e}$}
\newcommand{\pdist}{probability distribution}

\newcommand{\ket}[1]{\left\lvert #1 \right\rangle}

\newcommand{\density}[1]{\left\lvert #1 \right\rangle \left\langle #1 \right\rvert}

\begin{document}

\title{From estimation of quantum probabilities to simulation of quantum circuits}

\author{Hakop Pashayan}
\orcid{0000-0003-3782-7602}
\author{Stephen D. Bartlett}
\orcid{0000-0003-4387-670X}
\email{stephen.bartlett@sydney.edu.au}
\affiliation{Centre for Engineered Quantum Systems, School of Physics, The University of Sydney, Sydney, NSW 2006, Australia}

\author{David Gross}
\affiliation{Institute for Theoretical Physics, University of Cologne, Germany}

\date{24 December 2019}

\begin{abstract}

Investigating the classical simulability of quantum circuits provides a promising avenue towards understanding the computational power of quantum systems.  Whether a class of quantum circuits can be efficiently simulated with a probabilistic classical computer, or is provably hard to simulate, depends quite critically on the precise notion of ``classical simulation'' and in particular on the required accuracy. We argue that a notion of classical simulation, which we call \textsc{epsilon}-simulation (or \esim~for short), captures the essence of possessing ``equivalent computational power'' as the quantum system it simulates: It is statistically impossible to distinguish an agent with access to an $\epsilon$-simulator from one possessing the simulated quantum system.
We relate \esim~to various alternative notions of simulation predominantly focusing on a simulator we call a \emph{poly-box}. A poly-box outputs $1/poly$ precision additive estimates of Born probabilities and marginals. This notion of simulation has gained prominence through a number of recent simulability results. Accepting some plausible computational theoretic assumptions, we show that \esim~is strictly stronger than a poly-box by showing that IQP circuits and unconditioned magic-state injected Clifford circuits are both hard to $\epsilon$-simulate and yet admit a poly-box. In contrast, we also show that  these two notions are equivalent under an additional assumption on the sparsity of the output distribution (\emph{poly-sparsity}).
\end{abstract}

\maketitle

\section{Introduction and summary of main results}

Which quantum processes can be efficiently simulated using classical resources is a fundamental and longstanding problem~\cite{feynman1982simulating, aaronson2004improved, valiant2002quantum, terhal2002classical, bartlett2002efficient, gurvits2005complexity}.  Research in this area can be split into two broad classes: results showing the hardness of efficient classical simulation for certain quantum processes, and the development of efficient classical algorithms for simulating other quantum processes. Recently, there has been substantial activity on both sides of this subject. Works on boson sampling~\cite{aaronson2011computational}, instantaneous quantum polynomial (IQP) circuits~\cite{bremner2010classical, bremner2016average}, various translationally invariant spin models \cite{gao2017quantum, bermejo2017architectures}, quantum Fourier sampling \cite{fefferman2015power}, one clean qubit (also known as DQC1) circuits \cite{morimae2014hardness, morimae2017power}, chaotic quantum circuits \cite{boixo2018characterizing} and conjugated Clifford circuits \cite{bouland2017quantum} have focused on showing the difficulty of classically simulating these quantum circuits. On the other hand, there has been substantial recent progress in classically simulating various elements of quantum systems including matchgate circuits with generalized inputs and measurements~\cite{brod2016efficient} {{}}{(see also \cite{valiant2002quantum, terhal2002classical, Jozsa2008matchgates} for earlier works in this direction)}, circuits with positive quasi-probabilistic representations~\cite{veitch2012negative, veitch2013CV, mari2012positive}, stabilizer circuits supplemented with a small number of $T$ gates~\cite{bravyi2016improved}, stabilizer circuits with small coherent local errors~\cite{bennink2017unbiased}, noisy IQP circuits \cite{bremner2017achieving}, noisy boson sampling circuits \cite{oszmaniec2018classical}, low negativity magic state injection in the fault tolerant circuit model~\cite{howard2017application}, quantum circuits with polynomial bounded negativity \cite{pashayan2015estimating}, Abelian-group normalizer circuits~\cite{nest2012efficient, bermejo2014classical} and certain circuits with computationally tractable states and sparse output distributions \cite{schwarz2013simulating}. In addition, there has been some work on using small quantum systems to simulate larger quantum systems~\cite{bravyi2016trading} as well as using noisy quantum systems to simulate ideal ones~\cite{temme2017error}.

An important motivation for showing efficient classical simulability or hardness thereof for a given (possibly non-universal) quantum computer is understanding what properties of a quantum computer give rise to super-classical computational power.
 In this context, we desire classical simulability to imply that the computational power of the target quantum computer is ``contained in classical'', and the hardness of classical simulablility to imply that the target computational device can achieve at least some computational task beyond classical. Achieving these desiderata hinges crucially on the strength of the notion of simulation that is employed.
As an extreme example, if one uses a notion of simulation that is too weak, then efficient classical ``simulation'' of universal quantum circuits may be possible (even if BQP $\not \subseteq$ BPP). In such a case, the existence of a ``simulator'' does not imply that the computational power of the simulated system is contained within classical.
As an opposite extreme, if one uses a notion of simulation that is too strong, then efficient classical ``simulation'' of even classical circuits may be impossible \cite{nest2008classical}. In this case, the non-existence of such a simulator does not imply that the computational power of the ``un-simulable'' system is outside of classical. Once we establish the notion of simulation that is neither ``too strong'' nor ``too weak'', it will become evident that both too strong and too weak notions of simulations have been commonly used in the literature.
To this end, we require  a clear mathematical statement about which notion of simulation minimally preserves the computational power of the system it simulates.

{{}}{From a computer science perspective, the computational power of a device can be characterized by the set of problems such a device can solve. However, when it comes to quantum devices that produce probabilistic output from an exponentially growing space, even the question of what problems these devices solve or what constitutes a solution is subtle. Given an efficient description of a quantum circuit, the exact task performed by a quantum computer is to output a sample from the probability distribution associated with the measurement outcomes of that quantum circuit. This suggests that for ideal quantum computers, sampling from the \emph{exact} quantum distribution is what constitutes a solution.
On the other hand, it is unclear what well justified necessary requirement fail to be met by an arbitrarily small departure from exact sampling.
}
{{}}{Perhaps due to these subtleties, the choice of notion of ``classical simulation'' for sampling problems  lacks consensus and, under the umbrella term of \emph{weak simulation}, a number of different definitions have been used in the literature.}
{{}}{We will argue that some of these notions are too strong to be minimal and others are too weak to capture computational power.}
{{}}{The cornerstone of this argument will be the concept of \emph{efficient indistinguishability}; the ability of one agent to remain indistinguishable from another agent under the scrutiny of any interactive test performed by a computationally powerful referee whilst simultaneously employing resources that are polynomially equivalent.}

{{}}{Examples of definitions that we argue are too strong include simulators required to sample from exactly the target distribution or sample from a distribution that is exponentially close (in $L_1$-norm) to the target distribution \cite{jozsa2003role}. These also include a notion of simulation based on approximate sampling where the accuracy requirement is the very strong condition that every outcome probability is within a small \emph{relative} error of the target probability \cite{Terhal2004, bremner2010classical, morimae2017power}. From the perspective of efficient indistinguishibility, these notions of simulation are not minimal since they rule out weaker notions of simulation that are nonetheless efficiently indistinguishable from the target quantum system.}

{{}}{An example of a notion of approximate weak simulation that we argue is too weak requires that the classical algorithm sample from a distribution that is within some small fixed constant $L_1$-norm of the target distribution \cite{bremner2017achieving, oszmaniec2018classical, bremner2016average, gao2017quantum, bermejo2017architectures, morimae2017hardness, bouland2017quantum}. We argue that such a notion does not capture the full computational power of the target, since it cannot perform a task that can be performed by the target device, namely of passing some sufficiently powerful distinguishibility test.}
 
The focus of this paper will be on a notion of approximate weak simulation we call \emph{efficient polynomially small in $L_1$-norm} (\textsc{epsilon}) simulation (or \esim~for short). This has been used in prior works including Refs.~\cite{aaronson2011computational, aaronson2014equivalence, fefferman2015power, bravyi2016improved}. We will advocate for this notion of simulation (over other definitions of weak simulation) by showing that an \esimer~of a quantum computer achieves efficient indistinguishablity and any simulator that achieves efficient indistinguishablity satisfies the definition of an \esimer. Thus \esim~minimally captures computational power.
The notion of \esim~is also closely related to the definition of a sampling problem from Ref.~\cite{aaronson2014equivalence} (where the definition includes an exact statement of what constitutes a solution to the sampling problem). In this language, an \esimer~of a family of quantum circuits can be exactly defined as an efficient classical algorithm which can solve all sampling problems defined by the family of quantum circuits in the natural way. Thus, our result shows that a device can solve all sampling problems defined by a quantum computer if and only if the device is efficiently indistinguishable from the quantum computer.

{{}}{The conceptual significance of \esim~as a notion that minimally captures computational power motivates the study of its relation to other notions of simulation. This is particularly important for translating the existing results on simulability and hardness into statements about computational power relative to classical. Such a comparison to the above-mentioned approximate weak simulators is clear but a comparison to simulators defined in terms of Born probability estimation can be significantly more involved. Simulators which output sufficiently accurate Born probability estimates can be called as subroutines in an efficient classical procedure in order to output samples from a desired target distribution. Such a procedure can be used to ``lift'' these simulators to an \esimer~implying that the computational power of all families of quantum circuits simulable in this way is contained within classical.}

{{}}{Some commonly used notions of simulation such as \emph{strong simulation} and multiplicative precision simulation require the ability to estimate Born probabilities extremely accurately. These simulators can be lifted to \esimer s \cite{valiant2002quantum, terhal2002classical, Terhal2004}. We focus on another notions of simulation that has been prominent in recent literature \cite{pashayan2015estimating, howard2017application, bennink2017unbiased} which we call a \emph{poly-box}}
{{}}{. Compared to strong or multiplicative precision simulators, a poly-box has a much less stringent requirement on the accuracy of Born probability estimates that it produces. 
We discuss the significant conceptual importance of poly-boxes owing to the fact that they capture the computational power with respect to decision problems while simultaneously being weak enough to be admitted by IQP circuits,  unconditioned magic-state injected Clifford circuits and possibly other intermediate models for quantum computation.}

{{}}{Assuming some complexity theoretic conjectures, we show that a poly-box is a strictly weaker notion of simulation than \esim. However, if we impose a particular sparsity restriction on the target family of quantum circuits, then we show that a poly-box can be lifted to an \esimer, implying that the two notions are, up to efficient classical computation, equivalent under this sparsity restriction.}

\subsection{Outline of our main results}

\subsubsection{ {{}}{Indistinguishability and } \esim.}

In Sec.~\ref{simulation}, we  motivate the use of a particular notion of efficient simulation, which we call \esim. Essentially, we say that an algorithm can \emph{$\epsilon$-simulate} a family of quantum circuits, if for any $\epsilon>0$, it can sample from a distribution that is $\epsilon$-close in $L_1$-norm to the true output distribution of the circuit, and if the algorithm runs in time polynomial in $1/\epsilon$ and in the number of qubits.
We provide an operational meaning for this notion by showing that ``possessing an \esimer'' for a family of circuits is equivalent to demanding that even a computationally omnipotent referee cannot  distinguish the simulator's outputs from that of the target circuit family. Further, any simulator that satisfies efficient distinguishability also satisfies the definition of an \esimer. {{}}{This is captured by the following theorem presented in Sec.~\ref{simulation}.}

\begin{introthm}
{Bob has an \esimer~of Alice's quantum computer if and only if given the hypothesis testing scenario considered in Sec.~\ref{sec:motivating} there exists a strategy for Bob which jointly achieves indistinguishability and efficiency.}
 \end{introthm}

\subsubsection{Efficient outcome estimation:  the poly-box.}

{{}}{A family of binary outcome quantum circuits, where each circuit is indexed by a bit-string, defines a decision problem as follows: Given a bit-string indexing a quantum circuit, decide which of the circuit's two possible outcomes is more likely}\footnote{Technically, one is also promised that the given bit-string will only ever index a circuit where the probability of the two outcomes are bounded away from 50\%.}.

{{}}{Here, the only quantity relevant to the computation is the probability associated with the binary measurement outcome (decision). Hence, in this setting, simulation can be defined in terms of the accuracy to which these probabilities can be estimated.}
{{}}{A commonly used notion of simulation known as \emph{strong simulation} requires the ability to  estimate Born probabilities extremely accurately. In Sec.~\ref{estim}, we will define a much weaker notion of simulation} (\emph{poly-box}\footnote{This notion is similar to a notion introduced by Ref.~\cite{nest2009simulating} where it was (using a terminology inconsistent with the present paper) referred to as weak simulation.}) 
{{}}{which is a device that computes an additive polynomial precisions estimate of the quantum probability (or marginal probability) associated with a specific outcome of a quantum circuit.}

We show that families of quantum circuits must admit a poly-box in order to be \esimable.
\begin{introthm}
If $\mc{C}$ is a family of quantum circuits that does not admit a  poly-box algorithm, then $\mc{C}$ is not \esimable. 
\end{introthm}

{{}}{We advocate the importance of this notion on the grounds that whether or not some given family of quantum circuits admits a poly-box informs our knowledge of the computational power of that family relative to classical. In particular:} 
\begin{itemize}
	\item {{}}{if a (possibly non-universal) quantum computer can be efficiently classically simulated in the sense of a poly-box, then such a quantum computer cannot solve decision problems outside of classical}
	\item {{}}{if a (possibly non-universal) quantum computer cannot be efficiently classically simulated in the sense of a poly-box, then such a quantum computer can solve a sampling problem outside of classical (Thm.~\ref{necessity of polybox})}
\end{itemize}

We give three examples of poly-boxes.  
The first one is an estimator based on Monte Carlo sampling techniques applied to a quasiprobability representation. 
This follows the work of Ref.~\cite{pashayan2015estimating}, where the it was found that the efficiency of this estimator depends on the amount of ``negativity'' in the quasiprobability  description  of the quantum circuit.
{{}}{As a second example, we} consider the family of circuits \cmany, for which the $n$-qubit input state $\rho$ is an arbitrary product state (with potentially exponential negativity), transformations consist of Clifford unitary gates, and measurements are of $k\leq n$ qubits in the computational basis.  
We present an explicit  poly-box for \cmany~in Sec.~\ref{estim}.   
As a third example, we also outline a construction of a  poly-box for Instantaneous Quantum Polynomial-time (IQP) circuits \ciqp~based on the work of Ref.~\cite{shepherd2010binary}.

\subsubsection{From estimation to simulation.}

For the case of very high precision probability estimation algorithms, prior work has addressed the question of how to efficiently lift these to algorithms for high precision approximate weak simulators. In particular, Refs.~\cite{valiant2002quantum, terhal2002classical, Terhal2004} (see also Appendix \ref{strong sim proof}) lift estimation algorithms with small \emph{relative} error. In Appendix \ref{strong sim proof}, we also present a potentially useful algorithm for lifting small additive error estimators.
In Sec.~\ref{samp} we focus on the task of lifting an algorithm for a poly-box, to an \esimer. Since a poly-box is a much less precise probability estimation algorithm (in comparison to strong simulation),  achieving this task in the general case is implausible (see Sec.~\ref{hardness}). In Sec.~\ref{samp}, we will show that a poly-box can efficiently be lifted to an \esimer~if we restrict the family of quantum distributions to those possessing a property we call \emph{poly-sparsity}. This sparsity property measures ``peakedness versus uniformness'' of distributions and is related to the scaling of the smooth max-entropy of the output distributions of quantum circuits.  Loosely, a \emph{poly-sparse} quantum circuit can have its outcome probability distribution well approximated by specifying the probabilities associated with polynomially many of the most likely outcomes.  We formalize this notion in Sec.~\ref{samp}.  

\begin{introthm}
{Let $\mc{C}$ be a family of quantum circuits with a corresponding family of probability distributions $\mathbb{P}$. Suppose there exists a poly-box over $\mc{C}$, and that $\mathbb{P}$ is poly-sparse. Then, there exists an \esimer~of $\mc{C}$.}
\end{introthm}

We emphasize that the proof of this theorem is constructive, and allows for new simulation results for families of quantum circuits for which it was not previously known if they were efficiently simulable.  As an example, our results can be straightforwardly used to show that Clifford circuits with sparse outcome distributions and with small amounts of local unitary (non-Clifford) noise, as described in Ref.~\cite{bennink2017unbiased}, are \esimable.

\subsubsection{Hardness results.}

Finally, in Sec.~\ref{hardness}, we prove that the poly-box requirements of Theorem 2 is on its own not sufficient for $\epsilon$-simulability.  The challenge to proving such a result is identifying a natural family of non-poly-sparse quantum circuits for which a poly-box exists but for which \esim~is impossible.  

We prove that the family \cmany ~described above, which violates the poly-sparsity requirement, admits a poly-box.  Then, by assuming a now commonly used ``average case hardness'' conjecture~\cite{aaronson2011computational, bremner2016average, gao2017quantum, fefferman2015power, morimae2017hardness, bouland2017quantum}, we show that the ability to perform \esim~of \cmany ~implies the unlikely result that the polynomial hierarchy collapses to the third level.  Loosely, this result suggests that there exist quantum circuits where the probability of any individual outcome (and marginals) can be efficiently estimated, but the system cannot be \esimed.  Our hardness result closely follows the structure of several similar results, and in particular that of the IQP circuits result of Ref.~\cite{bremner2016average}.

Our proof relies on a conjecture regarding the hardness of estimating Born rule probabilities to within a small multiplicative factor for a substantial fraction of randomly chosen circuits from \cmany. This average case hardness conjecture (which we formulate explicitly as Conjecture~1) is a strengthening of the worst case hardness of multiplicative precision estimation of probabilities associated with circuits from \cmany. Worst case hardness {{}}{can be shown by applying the result of Refs.~\cite{Bravyi2005magic, jozsa2014classical, goldberg2017complexity, kuperberg2015hard, Fujii2017commuting}} and follows from an analogous argument to Theorem 5.1 of Ref.~\cite{bouland2017quantum}. 
	\begin{introthm}
		If there exists an \esimer~of \cmany ~and Conjecture~\ref{conj} holds, then the polynomial hierarchy collapses to the third level.
	\end{introthm}

{{}}{We note that our hardness result is implied by the hardness results presented in Refs.~\cite{gao2017quantum, bermejo2017architectures}, however; our proof is able to use a more plausible average case hardness conjecture than these references due to the fact that we are proving hardness of \esim ~rather than proving the hardness of the yet weaker notion of approximate weak simulation employed by these references. }

{{}}{In Appendix \ref{ps vs ac section} we also present Theorem \ref{ac vs ps thm}. This theorem shows that the properties of poly-sparsity and anti-concentration are mutually exclusive.}

The flow chart in Fig.~\ref{fig:FC} summarizes the main results in this paper by categorizing 
an{{}}{y} given family of quantum circuits in terms of its computation power based on whether or not the circuit family admits certain properties related to simulability.

  \begin{figure}[h!]
  \centering
	\includegraphics[width=1.0\textwidth, trim=30mm 70mm 35mm 45mm, clip]{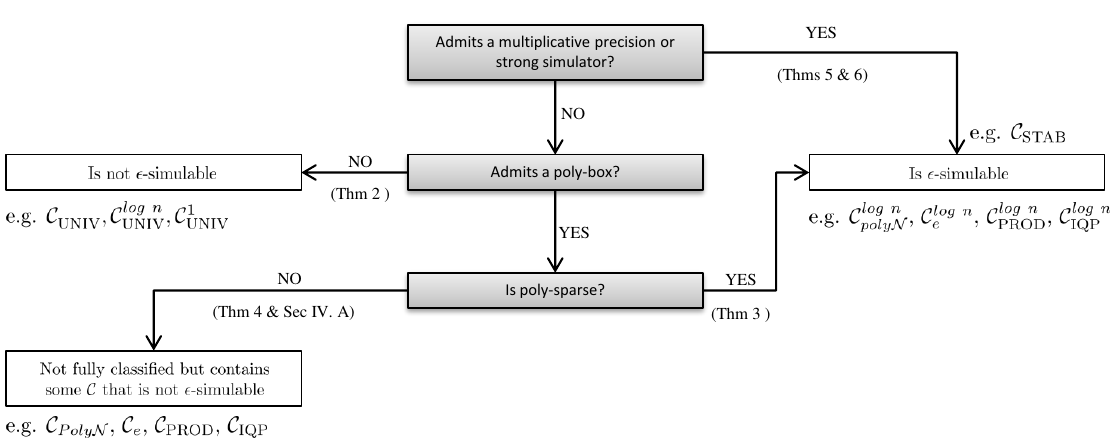}
  \caption{An overview of the main results. An arbitrary family of quantum circuits $\mc{C}$ is partially classified by its computational power relative to universal classical computers. The unclassified category (admits a poly-box and is not poly-sparse) is known to contain circuit families that are hard to $\epsilon$-simulate assuming some plausible complexity theoretic conjectures. We give examples of circuits families in these categories. Here, $\mc{C}^{*}_{\rm UNIV}, \mc{C}^{*}_{\rm STAB}, \mc{C}^{*}_{poly\mc{N}}$ and $\mc{C}^{*}_{\rm IQP}$ refer to the following families of circuits: universal circuits, stabilizer circuits, circuits with polynomially bounded negativity and IQP circuits respectively. The circuit families \cencode~and \cmany~are discussed in some detail in Sec.~\ref{polybox not suff} and \ref{defn cprod} respectively. The presence of a superscript represents an upper bound on the number of qubits to be measured.}
	\label{fig:FC}
  \end{figure}

\section{Defining simulation of a quantum computer}\label{simulation}

While there has been a breadth of recent results in the theory of simulation of quantum systems, this breadth has  been accompanied with a plethora of different notions of simulation.  This variety brings with it challenges for comparing results.  Consider the following results, which are all based on (often slightly) different notions of simulation. As a first example, the ability to perform \emph{strong simulation} of certain classes of quantum circuits would imply a collapse of the polynomial hierarchy, while under a weaker (but arguably more useful) notion of simulation this collapse is only implied if additional mathematical conjectures hold true~\cite{aaronson2011computational, bremner2016average}. As another example, Ref.~\cite{morimae2017power} shows that the quantum complexity class BQP is contained in the second level of the polynomial hierarchy if there exist efficient classical probabilistic algorithms for sampling a particular outcome (from the quantum circuits considered) {{}}{with a probability} that is exponentially close to the true quantum probability in terms of additive error (or polynomially close in terms of multiplicative error). As additional examples, Refs.~\cite{pashayan2015estimating, bennink2017unbiased} present efficient classical algorithms for additive polynomial precision estimates of Born rule probabilities. While many such technical results are crucially sensitive to these distinctions in the meaning of simulation, there is a growing need to connect the choice of simulation definition used in a proof against (or for) efficient classical simulability to a statement about proofs of quantum advantage (or ability to practically classically solve a quantumly solvable problem). In particular, to the non-expert it can be unclear what the complexity of classical simulation (in each of the above mentioned notions of simulation) of a given quantum device says about the hardness of building a classical device that can efficiently solve the computational problems that are solvable by the quantum device.

{{}}{In this section, we will discuss a meaningful notion of approximate weak simulation, which we call \esim.  This notion of simulation is a natural mathematical relaxation of exact weak simulation and has been used in prior works, e.g., in Refs.~\cite{aaronson2011computational, fefferman2015power, bravyi2016improved}. Further, this notion of simulation is closely related to the class of problems in complexity theory known as sampling problems \cite{aaronson2014equivalence}.  Here, we define \esim~and prove that up to polynomial equivalence, an \esimer~of a quantum computer is effectively a perfect substitute for any task that can be performed by the quantum computer itself. In particular, we will show that \esimer s satisfy \emph{efficient indistinguishability} meaning that they can remain statistically indistinguishable from (according to a computationally unbounded referee) and have a polynomially equivalent run-time to the quantum computer that they simulate. We argue that efficient indistinguishability is a natural choice of a rigorously defined global condition which minimally captures the concept of computational power. The accuracy requirements of \esim~are rigorously defined at the local level of each circuit and correspond to solving a sampling problem (as defined in \cite{aaronson2014equivalence}) based on the outcome distribution of the circuit. Thus our result shows that the ability to solve all sampling problems solvable by a quantum computer $\mathcal{C}$ is a necessary and sufficient condition to being efficiently indistinguishable from $\mathcal{C}$ or ``computationally as powerful as $\mathcal{C}$''.}

\subsection{Strong and weak simulation}\label{defn strong and weak sim} 

We note that every quantum circuit has an associated probability distribution that describes the statistics of the measurement outcomes. We will refer to this as the \emph{circuit's quantum probability distribution}. As an example, Fig.~\ref{fig:QC} below depicts a quantum circuit. The output of running this circuit is a classical random variable $X=(X_1, \ldots, X_k)$ that is distributed according to the quantum probability distribution.

  \begin{figure}[h!]
  \centering
  \includegraphics[width=.4\textwidth, trim=15mm 208mm 132mm 45mm, clip]{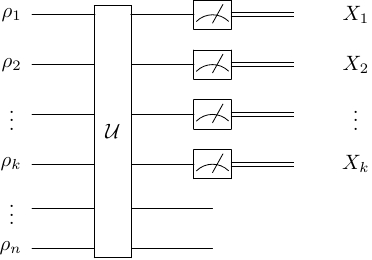}
  \caption{\label{fig:QC}An example of a quantum circuit.  This circuit acts on $n$ qubits (or, in general, qudits).  The initial state is a product state.  The unitary operation $\mathcal{U}$ must be constructed out of a sequence of local unitary gates.  The first $k$ qubits in this example are each measured in a fixed basis, yielding outcome $(X_1,X_2,\ldots,X_k)$.  Qubits $i>k$, shown without a measurement, are traced over (marginalized).}
  \end{figure}
  
Two commonly used notions of simulation are \emph{strong simulation} and \emph{weak simulation}. A weak simulator of a quantum circuit generates samples from the circuit's quantum probability distribution. In the strict sense of the term, a weak simulator generates samples from the \emph{exact} quantum probability distribution.  Loosely, having a weak simulator for a quantum system is an equivalent resource to using the quantum system itself.

The term \emph{weak simulation} has also been used in reference to classical algorithms which sample from distributions which \emph{approximate} the target probability distribution. There exist at least four distinct notions of approximate weak simulation appearing in the quantum computation literature. As background, we give a brief description of these here although the focus of this paper will be on only one of these and will be discussed in some detail later in this sections.
\begin{enumerate}
	\item The first notion of approximate weak simulation requires that the classical algorithm sample from a distribution that is exponentially close (in $L_1$-norm) to the target distribution.  This notion was used in Ref.~\cite{jozsa2003role, nest2008classical}.
	\item Another notion of apprximate weak simulation requires that the sampled distribution be sufficiently close to the target distribution so as to ensure that for every outcome $x$, the sampled distribution satisfies ${\abs{P_{sampled}(x)-P_{target}(x)}\leq \epsilon P_{target}(x)}$ for some fixed $\epsilon>0$. See Ref.~\cite{Terhal2004} and also \cite{bremner2010classical, morimae2017power} for related variants.
	\item The third notion of approximate weak simulation requires that the classical algorithm sample from a distribution that is inverse polynomially close (in $L_1$-norm) to the target distribution. This notion of simulation has been used in prior works, e.g., in Refs.~\cite{aaronson2011computational, aaronson2014equivalence, fefferman2015power, bravyi2016improved} both in the context of hardness of classical simulation and existence of classical simulators. We call this \esim.
	
	\item The final prominent example of approximate weak simulation, requires that the classical algorithm sample from a distribution that is within some small fixed constant $L_1$-norm of the target distribution. This definition has predominantly featured in hardness proofs \cite{bremner2016average, gao2017quantum, bermejo2017architectures, morimae2017hardness, bouland2017quantum}.
	It has also feature in proofs of efficient classical simulability of noisy boson sampling circuits \cite{oszmaniec2018classical} and noisy IQP circuits \cite{bremner2017achieving}. 
\end{enumerate}

A strong simulator, in contrast, outputs probabilities or marginal probabilities associated with the quantum distributions. More specifically, a strong simulator of a circuit is a device that outputs the quantum probability of observing any particular outcome or the quantum probability of an outcome marginalized\footnote{{{}}{A distribution $P(x)$ over  bit-strings $x\in \set{0,1}^n$ is said to have a marginal distribution $P_{\set{i_1,\ldots,i_m}}(\tilde{x})=\sum_{x_{i_1}} \ldots \sum_{x_{i_m}} P(x)$ (marginalized over the bits $\set{i_1,\ldots,i_m}$ where $\tilde{x}\in \bitstring{n-m}$ is given by modifying the vector $x$ by removing the entries $\set{i_1, \ldots, i_m}$.}} over one or more of the measurements.   Note that a strong simulator requires an input specifying the event for which the probability of occurrence is required. Taking Fig.~\ref{fig:QC} as an example, a strong simulator could be asked to return the probability of observing the event $(X_1,X_2)=(1,0)$, marginalized over the measurements $3$ to $k$.  The requirement that a strong simulator can also output estimates of marginals is weaker than requiring them to estimate the quantum probability associated with any event (subset of the outcome space).

While the names `strong' and `weak' simulation suggest that they are in some sense different magnitudes of the same type of thing, we note that these two types of simulation produce different types of output. In particular, a strong simulator outputs probabilities.  (More specifically, it outputs exponential additive precision estimates of Born rule probabilities and their marginals.)  In contrast a weak simulator outputs samples (from the exact target probability distribution).

{{}}{Ref.~\cite{nest2008classical} provides a compelling argument advocating for the use of weak simulation in place of strong simulation by showing that there exist  classically efficiently weak simulable probability distributions that are \#P-hard to strong simulate, thus showing that aiming to classically strong simulate is an unnecessarily challenging goal. In a similar vein, here we will advocate for the notion of \esim~ over other notions of simulation including the alternative notions of approximate weak simulation.}

\subsection{\esim}

A weak simulator, which generates samples from the exact quantum probability distribution, is a very strict notion.  Often, it would be sufficient to consider a simulator that generates samples from a distribution that is only sufficiently close to the quantum distribution, for some suitable measure of closeness.  Such a relaxation of the requirement of weak simulation has been used by several authors, e.g., in Refs.~ {{}}{\cite{jozsa2003role, nest2008classical,aaronson2011computational, aaronson2014equivalence, fefferman2015power, bravyi2016improved,bremner2016average, gao2017quantum, morimae2017hardness, bouland2017quantum, oszmaniec2018classical,bremner2017achieving}
}.  Here, we define the notion of \esim , which is a particular relaxation of the notion of weak simulation, and motivate its use.

We first define a notion of sampling from a distribution that is only \emph{close} to a given distribution.  Consider a discrete probability distribution $\mc{P}$.  Let $B(\mc{P},\epsilon)$ denote the $\epsilon$ ball around the target $\mc{P}$ according to the $L_1$ distance (or equivalently, up to an irrelevant constant, the total variation distance).  We define $\epsilon$-sampling of a probability distribution $\mc{P}$ as follows:

\defn{Let $\mc{P}$ be a discrete probability distribution. We say that a classical device or algorithm can \emph{$\epsilon$-sample} $\mc{P}$ iff for any $\epsilon >0$, it can sample from a probability distribution $\mc{P}^{\epsilon} \in B(\mc{P},\epsilon)$.  In addition, its run-time should scale at most polynomially in $1/\epsilon$.}

We note that the use of the $L_1$-norm in the above is motivated by the fact that the  {{}}{$L_1$-distance upper bounds on the one-shot success probability of distinguishing between two distributions}. More details can be found in the proof of Theorem \ref{indistinguishability thm} in Appendix \ref{indistinguishability proof}.

The definition above does not require the device to sample from precisely the quantum \pdist~$\mc{P}$, but rather allows it to sample from any \pdist~ $\mc{P}^{\epsilon}$ which is in the $\epsilon$ ball around the target \pdist, $\mc{P}$.  We note that the device or algorithm will in general take time (or other resources) that depends on the desired precision $\epsilon$ in order to output a sample, hence the efficiency requirement ensures that these resources scale at most polynomially in the precision $1/\epsilon$.

\defn{We say that a classical device or algorithm can $\epsilon$-simulate a quantum circuit if it can $\epsilon$-sample from the circuit's associated output probability distribution $\mc{P}$.}

We note that each of the above mentioned notions of simulation refers to the simulation of a single quantum circuit. More generally, we may be interested in (strong, weak, or $\epsilon$) simulators of uniform families of quantum circuits. In this setting we can discuss the efficiency of a simulator with respect to $n$, the number of qubits\footnote{As a technical condition, we require the circuit size, run-time (or any other resources) as well as the length of the circuit's description to be upper-bounded by $poly(n)$.}. As an example, consider a family of circuits described by a mapping from $\mc{A}^*$ (finite strings over some finite alphabet $\mc{A}$) to some set of quantum circuits
$\mc{C}=\set{c_{a}~|~a\in \mc{A}^*}$ where for each $a \in \mc{A}^*$, $c_{a}$ is a quantum circuit with some efficient description\footnote{Such a description must satisfy the uniformity condition. This can be done by fixing a finite gate set, input state and measurement basis and explicitly defining an efficiently computable mapping between $\mc{A}^*$ and the gate sequence.} given by the index $a$. In the case of strong (weak) simulation, we say that a device can efficiently strong (weak) simulate the family of quantum circuits $\mc{C}$ if the resources required by the device to strong (weak) simulate $c_{a} \in \mc{C}$ are upper-bounded by a polynomial in $n$.  In the case of \esim, we require that the simulator be able to sample a distribution within $\epsilon$ distance of the quantum distribution efficiently in both $n$ and $1/{\epsilon}$.

\defn{We say that a classical device or algorithm can $\epsilon$-simulate a uniform family of quantum circuit $\mc{C}$ if for all $\epsilon>0$ and for any $c\in \mc{C}$ (with number of qubits $n$ and quantum distribution $\mc{P}$) it can sample from a probability distribution $\mc{P}^{\epsilon}\in B(\mc{P},\epsilon)$ in run-time $O(poly(n,\frac{1}{\epsilon}))$.}

\subsection{{{}}{\esim~and efficient indistinguishability}}\label{sec:motivating}

As noted earlier, this definition ensures that \esim~is a weaker form of simulation than exact weak simulation.  However, we point out that the notion of exact sampling may be weakened in a number of ways, with the \esim~approach being well suited to many applications related to quantum simulators. As an example, if the definition of simulation allowed for a fixed but small amount of deviation in $L_1$ distance (as opposed to one that can be made arbitrarily small) then computational power of a simulator will immediately be detectably compromised. The above notion of \esim~requires a polynomial scaling between the precision ($1/{\epsilon}$) of the approximate sampling and the time taken to produce a sample. Below (Theorem \ref{indistinguishability thm}), we will use a statistical indistinguishability argument to show that  a polynomial scaling is precisely what should be demanded from a simulator. In particular, we will show that a run-time which scales sub-polynomially in $1/\epsilon$ puts unnecessarily strong demands on a simulator while a super-polynomial run-time would allow the simulator's output to be statistically distinguishable from the output of the device it simulates.

We now introduce the hypothesis testing scenario we consider.

\begin{test}
Suppose Alice possesses a quantum computer capable of running a (possibly non-universal) family of quantum circuits $\mathcal{C}$, and Bob has some simulation scheme for $\mathcal{C}$ (whether it's an \esimer ~is to be decided).  Further, suppose that a referee with unbounded computational power {{}}{and with full knowledge of the specifications of $\mathcal{C}$,} will request data from either Alice or Bob and run a test that aims to decide between the hypotheses:
{

$H_a$: \emph{The requested data came from Alice's quantum computer} or 

$H_b$: \emph{The requested data came from Bob's simulator}. 
}

The setup will be as follows: At the start of the test, one of Alice or Bob will be randomly appointed as ``the candidate''.  Without knowing their identity, the referee will then enter into a finite length interactive protocol with the candidate (see Fig \ref{interactive protocol}).
Each round of the protocol will involve the referee sending a circuit description to the candidate requesting the candidate to run the circuit and return the outcome. The choice of requests by the referee may depend on all prior requests and data returned by the candidate. The rules by which the referee:
\begin{enumerate}
	\item chooses the circuit requested in each round,
	\item chooses to stop making further circuit requests and
	\item decides on $H_a$ versus $H_b$ given the collected data
\end{enumerate}
define \emph{the hypothesis test}.
 The goal of the referee is as follows. For any given $\delta>0$ decide $H_a$ versus $H_b$ such that $P_{correct}\geq \frac{1}{2}+\delta$ where $P_{correct}$ is the probability of deciding correctly. Bob's goal is to come up with a ($\delta$-dependent) strategy for responding to the referee's requests such that it jointly achieves:
\begin{itemize}
	\item \emph{indistinguishablity:} for any $\delta>0$ and for any test that the referee applies, $P_{correct}< \frac{1}{2}+\delta$ and
	\item \emph{efficiency:} for every choice of circuit request sequence $\alpha$, Bob must be able to execute his strategy using resources which are $O(poly(N(\alpha), \frac{1}{\delta}))$ where $N(\alpha)$ is the resource cost incurred by Alice for the same circuit request sequence.
\end{itemize}
  
\end{test}

We note that the referee can always achieve a success probability $P_{correct}=\frac{1}{2}$ simply by randomly guessing $H_a$ or $H_b$. {{}}{Importantly, the referee has complete control over the number of rounds in the test and additionally does not have any upper bound imposed on the number of rounds. Hence, $P_{correct}$ is the ultimate one shot probability of the referee correctly deciding between $H_a$ or $H_b$ and in no sense can this probability be amplified through more rounds of information requests. As such, we will say that the referee achieves distinguishability between Alice and Bob if $\forall \delta>0$, there exists a test that the referee can apply ensuring that $P_{correct}\geq 1-\delta$ (independent of Bob's strategy). Alternatively, we will say that Bob achieves indistinguishability (from Alice) if $\forall \delta>0$, there exists a response strategy for Bob such that  $P_{correct}\leq \frac{1}{2}+\delta$ (independent of what test the referee can apply).} 
We will show that if Bob has an \esimer ~then there exists a strategy for Bob such that he jointly achieves indistinguishablity (i.e. the referee cannot improve on a random guess by any fixed probability $\delta>0$) and efficiency. {{}}{In this case, Bob can at the outset choose any $\delta>0$ and ensure that $P_{correct}< \frac{1}{2}+\delta$ for all strategies the referee can employ.}

The efficiency requirement imposed on Bob's strategy is with respect to the resource cost incurred by Alice. Here we will define what this means and justify the rationale behind this requirement. Let us first note that for any circuit $c_a \in \mc{C}$, there are resource costs $R(c_a)$ incurred by Alice in order to run this circuit. This may be defined by any quantity as long as this quantity is upper and lower-bounded by some 
polynomial in the number of qubits. For example, $R(c_a)$ may be defined by run-time, number of qubits, number of elementary gates, number of qubits plus gates plus measurement, length of circuit description etc. Since this quantity is polynomially equivalent to the number of qubits, without loss of generality, we can treat $n_a$ (the number of qubits used in circuit $c_a$) as the measure of Alice's resource cost $R(c_a)$. We now note that for a given test, the referee may request outcome data from some string of circuits $c_1, \ldots, c_m \in \mc{C}$. Thus we define the resource cost for Alice to meet this request by $N:=n_1+\ldots +n_m$. 

Bob's resource cost (run-time) with respect to each circuit $c_a\in \mc{C}$ is polynomially dependent on both $n_a$ and the inverse of his choice of accuracy parameter $\epsilon$. Thus, Bob's strategy is defined by the rules by which he chooses $\epsilon_j$, the accuracy parameter for his response in the $j^{th}$ round\footnote{Bob must posses{{}}{s} some computational power in order to execute these rules. We will only require that Bob have some small amount of memory (to keep count of the rounds in the protocol) and compute  simple arithmetic functions of this counter.}. Thus, for a given sequence of circuit requests $a_1, \ldots, a_m \in \mc{A}^*$, Bob will incur a resource cost $T=t_1+\ldots +t_m$ where $t_j\sim poly(n_{a_j}, 1/\epsilon_j)$ is Bob's resource in the $j^{th}$ round. Thus the efficiency condition requires that there exists some polynomial $f(x, y)$ such that for all $\delta>0$ and for all possible request sequences $\alpha=(a_1, \ldots, a_m)$, $T(\alpha)\leq f(N(\alpha), \frac{1}{\delta})$. 
The efficiency requirement imposed on Bob's strategy thus naturally requires that the resource costs of Alice and Bob be polynomial equivalent for the family of tests that the referee can apply.

\begin{thm}\label{indistinguishability thm}
{Bob has an \esimer ~of Alice's quantum computer if and only if given the hypothesis testing scenario considered above, there exists a strategy for Bob which jointly achieves indistinguishablity and efficiency.}
\end{thm}

The proof for this theorem can be found in Appendix \ref{indistinguishability proof}. 
{{}}{The proof uses the fact that the $L_1$ distance between Alice and Bob's output distributions over the entire interactive protocol can be used to upper bound the  probability of correctly deciding between $H_a$ and $H_b$. Further, we show that the total $L_1$ distance between Alice and Bob's output distributions over the entire interactive protocol grows at most additively in the $L_1$ distance of each round of the protocol. We also note that an \esimer~allows Bob to ensure that the $L_1$ distance of each round decays like an inverse quadratic ensuring that the sum of the $L_1$ distances converges to the desired upper bound. The convergence of the inverse quadratic series, which is an inverse polynomial, thus motivates the significance of \esimer s i.e. simulators with run-time $O(poly(n,1/\epsilon))$.}

{{}}{We note that the ``if'' component of the theorem says that meeting the definition of \esimer ~is necessary for achieving efficient indistinguishability, thus the notion of simulation cannot be weakened any further without compromising efficient indistinguishability.}

{{}}{Throughout this paper we view a quantum computer as a uniform family of quantum circuits $\mc{C}=\set{c_{a}~|~a\in \mc{A}^*}$. We note that by committing to the circuit model of quantum computation, our language including important definitions such as \esim~are not necessarily well suited to other models of computation unless these are first translated to the circuit model.  For example, in a quantum computational model that makes use of intermediate measurements, such as the measurement based quantum computing (MBQC) model, consider a procedure where a part of the state is measured then conditioned on the outcome, a second measurement is conducted. This procedure (consisting of 2 rounds of measurement) can be described as a single circuit in the circuit model, but cannot be broken up into two rounds involving two separate circuits.  This limitation becomes apparent when we consider the hypothesis testing scenario.  If the referee is performing a multi-round query, expecting the candidate to possess an MBQC-based quantum computer, then even Alice with a quantum computer may be unable to pass the test unless her computer operates in an architecture that can maintain quantum coherence between rounds.  In the setting we consider, such a query by the referee in not allowed.}

\subsection{{{}}{\esim~and computational power}}\label{sec:esim and comp power}

{{}}{In addition to the technical contribution of Theorem \ref{indistinguishability thm}, we wish to make an argument for the conceptual connection between computational power and efficient indistinguishability. Intuitively, we wish to say that an agent $A$ is at least as computationally powerful as agent $B$ if $A$ can ``do'' every task that $B$ can do using an equivalent amount of resources. In our setting, we can restrict ourselves to polynomially equivalent resources and the most general task of sampling from a target probability distribution given an efficient description of it. However, defining what constitutes an acceptable solution to the sampling task is not only of central importance but also difficult to conceptually motivate. Given a description of a probability distribution, can anything short of sampling \emph{exactly} from the specified distribution constitute success? An answer in the negative seems unsatisfactory because very small deviations}\footnote{For example consider the scenario that whenever an agent is asked to sample from some distribution $\mathcal{P}$, they output samples from exactly $\mathcal{P}$ every time, possibly with one exception. In particular, a memory bit stores if the exception has ever taken place. If it has occurred, then forever more, when the agent is asked to sample from $\mathcal{P}$, an exact sample is produced. If the exception has not yet taken place then with some very small probability, the agent will output the most likely outcome instead of an exact sample from $\mathcal{P}$.}
 {{}}{from exact sampling are ruled out. However, an answer in the positive presents the subtlety of specifying the exact requirement for achieving the task. 
It is easy to offer mathematically reasonable requirements for what constitutes success at the local level of each task but significantly more difficult to conceptually justify these as precisely the right notion. In our view, this difficulty arises because a well formed conceptually motivated requirement at the local level of each task must be inherited from a global requirement imposed at the level of the agent across their performance on any possible task.}

{{}}{We advocate for efficient indistinguishability as the right choice of global requirement for defining computational power and implicitly defining what constitutes a solution to a sampling task. If an agent is efficiently indistinguishable from another then, for any choice of $\delta>0$ chosen at the outset, the referee cannot assign any computational task to the candidate to observe a consequence that will improve (over randomly guessing) their ability to correctly decide between $H_a$ and $H_b$ by a probability $\delta$.
Thus, there is no observable consequence}\footnote{{{}}{Since the observer is the computationally unbounded referee, then any event is an observable consequence i.e. if we let $S$ be the set of all possible responses across all rounds from both Alice and Bob, then any element of the power set of $S$ is an observable consequence.}}{{}}{ to substituting an agent with another efficiently indistinguishable agent. For these reasons, we argue that in the setting where the agents are being used as computational resources, an agent's ability to (efficiently and indistinguishably) substitute another naturally defines containment of computational power.
In light of this, the ``only if'' component of Theorem \ref{indistinguishability thm} says that, the computational power of Bob (given an \esimer ~of $\mc{C}$) contains that of Alice (given $\mc{C}$) and the ``if'' component says that an \esimer~is the \emph{minimal} simulator that achieves this since any simulator to achieve efficient indistinguishibility is an \esimer.}

{{}}{The referee can be seen as a mathematical tool for bounding the adversarial ability of any natural process to distinguish an agent from an efficiently indistinguishable substitute. As such one may argue for further generalization of the concept of efficient indistinguishability from one which is defined with respect to (w.r.t.) a computationally unbounded referee to a notion dependent on the computational power of the referee. If we take the view that the computational power of all agents within this universe is bounded by universal quantum computation, then a particularly interest generalization is efficiently indistinguishability w.r.t. a referee limited to universal quantum computation. We return to this generalization in the discussion, elsewhere focusing on efficient indistinguishability w.r.t. a computationally unbounded referee.}

\section{Probability Estimation}\label{estim}

As described in the previous section, an exact (or approximate in the sense of Ref.~\cite{jozsa2003role}) weak simulator produces outcomes sampled from the exact {{}}{(or exponentially close to the exact)} Born rule probability distribution associated with a quantum circuit. {{}}{The notion of \esim~is a weaker notion of simulation, a fact we aim to exploit by constructing algorithms for \esim~ that would not satisfy the above-mentioned stronger notions of simulation.}
  In this paper, we describe an approach to \esim~of quantum circuits based on two components:  first, estimating Born rule probabilities for specific outcomes of a quantum circuit to a specified precision, and then using such estimates to construct a simulator.  In this section, we describe this first component, coined a \emph{poly-box}. {{}}{We motivate and define poly-boxes, discuss their conceptual importance and give a number of important examples.} In the next section, we employ such an estimator to construct an \esimer~under certain conditions.

\subsection{{{}}{Born rule probabilities and estimators}}

Consider the description $c=\set{\rho, U, \mathcal{M}}$ of some ideal quantum circuit, with $\rho$ an initial state, $U=U_L U_{L-1} \cdots  U_1$ a sequence of unitary gates, and $\mathcal{M}$ a set of measurement operators (e.g., projectors).  Associated with the measurement $\mathcal{M}=\{E_{x}~|~x\in \bitstring{k}\}$ is a set of possible measurement outcomes $x\in \bitstring{k}$. The Born rule gives us the exact quantum predictions associated with observing any particular outcome $x$: 
\begin{align}
	\mc{P}(x):=\tr(U\rho U^{\dagger} E_{x})\,. \label{Born prob outcomes}
\end{align}
Further, probabilities associated with events $S\subseteq \bitstring{k}$ are given by:
\begin{align}
	\mc{P}(S):=\sum_{x\in S} \tr(U\rho U^{\dagger} E_{x}) \label{Born prob events}
\end{align}

{{}}{The task of efficiently classically estimating these probabilities with respect to general quantum circuits is of great practical interest, but is known to be hard  even for rather inaccurate levels of estimation.} For example, given a circuit $c_{a}$ from a family of universal quantum circuits with a Pauli $Z$ measurement of the first qubit only,
deciding if $\mc{P}_{a}(0)>\frac{2}{3}$ or $<\frac{1}{3}$ is BQP-complete.

Monte Carlo methods are a common approach to estimating Born rule probabilities that are difficult to calculate directly \cite{pashayan2015estimating, howard2017application, bennink2017unbiased}.  Let $p$ be an unknown parameter we wish to estimate, e.g., a Born rule probability.  In a Monte Carlo approach, $p$ is estimated by observing a number of random variables $X_1, \ldots, X_s$ and computing some function of the outcomes $\hat{p}_s(X_1, \ldots, X_s)$, chosen so that $\hat{p}_s$ is close to $p$ in expectation. In this case, $\hat{p}_s$ is an estimator of $p$.

We first fix some terminology regarding the precision of as estimator, and how this precision scales with resources.  We say that an estimator $\hat{p}_s$ of $p$ is \emph{additive $(\epsilon, \delta)$-precision} if:
\begin{align}
	\pr\bigl(\abs{p-\hat{p}_s}\geq \epsilon\bigr)\leq \delta \,, \qquad \text{additive $(\epsilon, \delta)$-precision.}
\end{align}
We say that $\hat{p}_s$ is \emph{multiplicative $(\epsilon, \delta)$-precision} if:
\begin{align}
	\pr\bigl(\abs{p-\hat{p}_s}\geq \epsilon p\bigr)\leq \delta \,, \qquad \text{multiplicative $(\epsilon, \delta)$-precision.}
\end{align}
In the case where $p\leq 1$ is a probability, a multiplicative precision estimator is more accurate than an additive precision estimator. 

For any estimator based on the Monte Carlo type of approach described above, there is a polynomial (typically linear) resource cost associated with the number of samples $s$.  For example, the time taken to compute $\hat{p}_s$ will scale polynomially in $s$. More generally, $s$ may represent some resource invested in computing the estimator $\hat{p}_s$ such as the computation run-time. For this reason, we may wish to classify additive/multiplicative $(\epsilon, \delta)$-precision estimators by how $s$ scales with $1/{\epsilon}$ and $1/{\delta}$.  We say that $\hat{p}_s$ is an additive \emph{polynomial precision estimator} of $p$ if there exists a polynomial $f(x,y)$ such that for all $\epsilon,\delta>0$, $\hat{p}_s$ is an additive $(\epsilon, \delta)$-precision estimator for all $s\geq f(\epsilon^{-1}, \log\delta^{-1})$.  We say that $\hat{p}_s$ is a multiplicative \emph{polynomial precision estimator} of $p$ if there exists a polynomial $f(x,y)$ such that for all $\epsilon,\delta>0$, $\hat{p}_s$ is a multiplicative $(\epsilon, \delta)$-precision estimator for all $s\geq f(\epsilon^{-1}, \delta^{-1})$ \footnote{{{}}{The observant reader will notice that additive and multiplicative precision estimators have different scalings in $\delta$. Of course one can define an alternative notion of additive estimation where $s\geq f(\epsilon^{-1}, \delta^{-1})$ or an alternative notion of multiplicative estimation where $s\geq f(\epsilon^{-1}, \log~\delta^{-1})$. Here, we have chosen to define the notions that are most useful as motivated by the existence of techniques and associated inequalities bounding their performance. In particular, Hoeffding's inequality allows the construction of additive $(\epsilon, \delta)$-precision estimators while Chebyshev's inequality motivates the multiplicative $(\epsilon, \delta)$-precision estimator definition.}}.

A useful class of polynomial additive precision estimators is given by application of the Hoeffding inequality. Suppose $\hat{p}_1$ resides in some interval $[a,b]$ and is an unbiased estimator of $p$ {{}}{(i.e.} $\mathbb{E}(\hat{p}_1)=p${{}}{)}. Let $\hat{p}_s$ be defined as the average of $s$ independent observations of $\hat{p}_1$. Then, by the Hoeffding inequality, we have:
\begin{align}
	\pr\bigl(\abs{p-\hat{p}_s}\geq \epsilon\bigr)\leq 2 \exp \Bigl({\frac{-2s \epsilon^2}{(b-a)^2}}\Bigr)\,,
\end{align}
for all $\epsilon>0$. We note that for $s(\epsilon^{-1}, \log\delta^{-1})\geq \frac{(b-a)^2}{2 \epsilon^2}\log(2\delta^{-1})$, $\hat{p}_s$ is an additive $(\epsilon, \delta)$-precision estimator of $p$.  With this observation, we see that additive polynomial precision estimators can always be constructed from unbiased estimators residing in a bounded interval.

{{}}{As an important example let us consider one way an agent can generate Born probability estimates when given access to some classical processing power and a family of quantum circuits $\mathcal{C}$.  Given a description of an event $S$ and a description of a quantum circuit $c_{a}\in \mathcal{C}$, the agent can efficiently estimate $p=\mc{P}_{a}(S)$.  In this example, the agent can construct the estimator $\hat{p}_s$ by independently running the circuit $s$ times.  On each of the runs $i=1, \ldots, s$, she observes if the outcome $x$ is in the event $S$ (in this case, $X_i=1$) or not in $S$ (in this case, $X_i=0$). We then define $\hat{p}_s=\frac{1}{s}\sum_{i=1}^s X_i$. Using the Hoeffding inequality, it is easy to show that the Born rule probability estimator $\hat{p}_s$ is an additive polynomial precision estimator of $p$. Thus, for all $a\in \mc{A}^*, \epsilon, \delta>0$, there is a choice of $s\in \bbn$ such that this procedure can be used to compute an estimate $\hat{p}$ of $p:=\mc{P}_{a}(S)$ such that $\hat{p}$ satisfies the accuracy requirement:}
	\begin{align}
		\pr\bigl(\abs{p-\hat{p}}\geq \epsilon\bigr)\leq \delta\label{accuracy req PB}
	\end{align}
{{}}{and the run-time required to compute the estimate $\hat{p}$ is $O(poly(n,\epsilon^{-1},\log~\delta^{-1}))$.}

{{}}{Let us now discuss an important aspect that we have been ignoring: namely the restrictions that need to be placed on the events $S$. We first note that since each event $S$ is an element of the power set of $\bitstring{k}$, the total number of events grows doubly exponentially implying that any polynomial length description of events can only index a tiny fraction of the set of all events. Even once we make a particular choice as to how (and hence which) events are indexed by polynomial length descriptions, deciding if a bit-string $x$ is in the event $S$ is not computationally trivial (with the complexity depending on the indexing). Since the estimation procedure requires a computational step where the agent checks whether $x$ is in $S$, there will be restrictions place on the allowed events depending on the computational limitations of the agent and the complexity of the indexing of events.}

{{}}{
When discussing poly-boxes, we will be interested in the restricted set of events $S\in \set{0,1,\bullet}^k$. We use this notation to indicate the set of all specific outcomes and marginals. Specifically, $S\in  \set{0,1,\bullet}^{k}$ is a subset of $\bitstring{k}$ where $\bullet$ is a ``wild card'' single qubit measurement outcome and hence is consistent with both a $0$ and a $1$ element. For example, $S=(0,\bullet,1):=\set{(0,0,1), (0,1,1)}$. If $S$ is a vector with bits $x_1, \ldots, x_{k-m}$ in positions $i_1, \ldots, i_{k-m}$ and a $\bullet$ in each of the positions $j_1, \ldots, j_{m}$; then $S$ represents the set of $2^m$ outcomes where the qubits numbered $i_1, \ldots, i_{k-m}$ produced single qubit measurement outcomes corresponding to $x_1, \ldots, x_{k-m}$ while the remaining qubits (those numbered $j_1, \ldots, j_{m}$) produce either a $0$ or a $1$ measurement outcome. The probability corresponding to such an event $S$ is the marginal probability associated with observing the outcome bit-string $x_1, \ldots, x_{k-m}$ on the qubits numbered $i_1, \ldots, i_{k-m}$ marginalized over the qubits $j_1, \ldots, j_m$.}

\subsection{{{}}{The poly-box:  generating an additive polynomial precision estimate}}

Given a family of quantum circuits $\mc{C}$, we will be interested in constructing an \esimer ~of $\mc{C}$ using estimates of Born rule probabilities associated with circuits in $\mc{C}$. For this purpose  {{}}{we} define a poly-box over $\mc{C}$.

\defn{\emph{(poly-box).}  A poly-box over a family of quantum circuits $\mc{C}=\set{c_a~|~a\in \mc{A}^*}$ with associated family of probability distributions $\mathbb{P}=\set{\mc{P}_a~|~a\in \mc{A}^*}$ is a classical algorithm that, for all $a\in \mc{A}^*, \epsilon, \delta>0$ and $S\in \set{0,1,\bullet}^{k_n}$, can be used to compute an estimate $\hat{p}$ of $\mc{P}_{a}(S)$ such that $\hat{p}$ satisfies the accuracy requirement:
	\begin{align}
		\pr\bigl(\abs{p-\hat{p}}\geq \epsilon\bigr)\leq \delta\label{accuracy req PB}
	\end{align}
and, the run-time required to compute the estimate $\hat{p}$ is $O(poly(n,\epsilon^{-1}, \log~\delta^{-1}))$.
}

Eq.~\eqref{accuracy req PB}, gives an upper bound on the probability that the computed estimate, $\hat{p}$, is far from the target quantity. This probability is over the potential randomness in the process used to generate the estimate $\hat{p}$. In addition we implicitly assume that the output of this process is independent of prior output. In particular, let $\alpha=(a, \epsilon, \delta, S)$ be an input into a poly-box and $\hat{p}_{\alpha}$ the observed output. Then, we implicitly assume that the probability distribution of $\hat{p}_{\alpha}$ only depends on the choice of input $\alpha$ and in particular is independent of prior output.

Note that a poly-box over a family of quantum circuits $\mc{C}=\set{c_a~|~a\in \mc{A}^*}$ with associated family of probability distributions $\mathbb{P}=\set{\mc{P}_a~|~a\in \mc{A}^*}$ is a classical algorithm that can be used to compute additive polynomial precision estimators $\hat{p}_s$ of $\mc{P}_{a}(S)$ for all $a\in \mc{A}^*, s\in \bbn, S\in \set{0,1,\bullet}^{k_n}$ efficiently in $s$ and $n$.

\subsection{{{}}{Conceptual significance of a poly-box}}

Whether or not a family of quantum circuits $\mc{C}$ admits a poly-box has bearing on both the complexity of sampling problems and decision problems solvable by $\mc{C}$, and so we will find that the notion of a poly-box is a useful concept.  We {{}}{first note} that the existence of a poly-box is  a necessary condition for \esim.

\begin{thm}\label{necessity of polybox}
If $\mc{C}$ is a family of quantum circuits that does not admit a  poly-box algorithm, then $\mc{C}$ is not  \esimable. 
\end{thm}

\begin{proof}
We note that given an \esimer ~of $\mc{C}$, a poly-box over $\mc{C}$ can be constructed in the obvious way simply by observing the frequency with which the \esimer ~outputs outcomes in $S$ and using this observed frequency as the estimator for $\mc{P}(S)$.
\end{proof}

{{}}{A poly-box over $\mc{C}$ is not only necessary for the existence of an \esimer~over $\mc{C}$, but as we will show in Theorem  \ref{simable}, combined with an additional requirement, it is also sufficient.} In addition, we note that if $\mc{C}$ admits a poly-box then all ``generalized decision problems'' solvable by $\mc{C}$ are solvable within BPP. As an illustrative but unlikely example, suppose there exists a classical poly-box over a universal quantum circuit family \cuniv. Then, for any instance $x$ of a decision problem $L$ in BQP, there is a quantum circuit $c_a \in $ \cuniv~that decides if $x \in L$ (correctly on at least $2/3$ of the runs), simply by outputting the decision ``$x \in L$'' when the first qubit measurement outcome is $1$ on a single run of $c_a$ and conversely, outputting the decision ``$x \not \in L$'' when the first qubit measurement outcome is $0$. We note that, in order to decide if $x\in L$ one does not need the full power of an \esimer~over \cuniv. In fact it is sufficient to only have access to the poly-box over \cuniv. Given a poly-box over \cuniv, one can request an $(\epsilon, \delta)$-precision estimate $\hat{p}$ for the probability $p$ that the sampled outcome from $c_a$ is in $S=(1,\bullet, \ldots, \bullet)$. For $\epsilon < 1/6$ and $\delta <1/3$, one may decide ``$x\in L$'' if $\hat{p}\geq 1/2$ and ``$x \not \in L$'' otherwise. This will result in the correct decision with probability $\geq 2/3$ as required. A poly-box over $\mc{C}$ offers the freedom to choose any $S\in \set{0,1,\bullet}^n$ which can in general be used to define a broader class of decision problems. Of course in the case of \cuniv, this freedom cannot be exploited because for every choice of $a$ and $S\neq (1,\bullet, \ldots, \bullet)$, there is a alternative easily computable choice of $a'$ such that the probability that a run of $c_{a'}\in $ \cuniv~results in an outcome in $(1,\bullet, \ldots, \bullet)$ is identical to the probability that a run of $c_{a}\in $ \cuniv~results in an outcome in $S$. {{}}{However, since we are considering the general case of not necessarily universal families of quantum circuits, it is feasible that a poly-box over $\mathcal{C}$ will be computationally more powerful than a poly-box over $\mathcal{C}$ restricted to only estimating probabilities of events of the form $S= (1,\bullet, \ldots, \bullet)$. On the other hand, we do not wish to make poly-boxes exceedingly powerful. If we view a poly-box over $\mc{C}$ as a black box containing an agent with access to $\mc{C}$ and processing an estimation algorithm as per the aforementioned example, then by restricting the allowable events as above and choosing such a simple method of indexing these, we are able to limit the additional computational power given to the agent and/or poly-box.}

\subsection{Examples of poly-boxes}

\subsubsection{Poly-boxes from quasiprobability representations}

There are a number of known algorithms for constructing poly-boxes over certain non-universal families of quantum circuits~\cite{aaronson2004improved, stahlke2014quantum, pashayan2015estimating, bravyi2016improved, bennink2017unbiased}. In particular, we focus on the algorithm presented in Ref.~\cite{pashayan2015estimating}, which can be used to construct a poly-box over any family of quantum circuits $\mc{C}$ where the \emph{negativity} of quantum circuits grows at most polynomially in the circuit size. We refer the interested reader to Ref.~\cite{pashayan2015estimating} for a definition of the negativity of a quantum circuit, but note that this quantity depends on the initial state, sequence of unitaries and the final POVM measurement that defines the circuit.  For general quantum circuits, the negativity can grow exponentially in both the number of qudits and the depth of the circuit.

A key application of this approach is to Clifford circuits. In odd dimensions, stabilizer states, Clifford gates, and measurements in the computational basis do not contribute to the negativity\footnote{{{}}{with respect to either the phase point operator or  stabilizer states choice of frame}} of a Monte Carlo based estimator.  {{}}{I}ncluding product state preparations or measurements that are not stabilizer states, or non-Clifford gates such as the $T$ gate, may contribute to the negativity of the circuit.  Nonetheless, these non-Clifford operations can be accommodated within the poly-box provided that the total negativity is bounded polynomially.  In addition, a poly-box exists for such circuits even in the case where the negativity of the initial state, \emph{or} of the measurement, is exponential~\cite{pashayan2015estimating}.

\subsubsection{A poly-box over \cmany}\label{defn cprod}
As a nontrivial example of a class of Clifford circuits for which there exists a poly-box, consider the family of circuits \cmany. This family consists of quantum circuits with an $n$-qubit input state $\rho$ that is an arbitrary product state\footnote{As an additional technical requirement, we impose that the input product state is generated from $\ket{0}^{\otimes n}$ by the application of polynomially many gates from a universal single qubit gate set with algebraic entries.} (with potentially exponential Wigner function negativity~\cite{pashayan2015estimating} in the input state). The allowed transformations are non-adaptive Clifford unitary gates, and $k\leq n$ qubits are measured at the end of the circuit, in the computational basis.  Such a circuit family has been considered by Jozsa and Van den Nest~\cite{jozsa2014classical}, where it was referred to as INPROD, OUTMANY, NON-ADAPT. 
This circuit family will be discussed again in Sec.~\ref{hardness} where we will show the classical hardness of simulating this family according to another notion of simulation. 
Aaronson and Gottesman~\cite{aaronson2004improved} provide the essential details of a poly-box for this family of circuits; for completeness, we present an explicit poly-box for \cmany ~in the following lemma.

\begin{lemma}{
	A classical poly-box exists for the Clifford circuit family \cmany.}
\end{lemma}

\begin{proof}
Give an arbitrary circuit $c=\{\rho, U, \mc{M}\}\in$ \cmany ~and an event $S\in \set{0,1,\bullet}^n$ we construct an estimator $\hat{p}_s$ of the probability $\mc{P}(S)$ as follows:
\begin{enumerate}
	\item Let $\Pi=\otimes_{i=1}^n \Pi_i$ be the projector corresponding to $S$. Here, we set:
			\begin{equation}
						 \Pi_i = 
							\begin{cases}
									\frac{I+Z}{2} & \text{if the $i^{th}$ entry of $S$ is $0$} \\
									\frac{I-Z}{2} & \text{if the $i^{th}$ entry of $S$ is $1$} \\
									I & \text{if the $i^{th}$ entry of $S$ is $\bullet$}
							\end{cases}
					\end{equation}
	\item For each $i$ where the $i^{th}$ entry of $S$ is not $\bullet$, $\Pi_i=\frac{I\pm Z}{2}$. In these cases, define a local Pauli operator $P_i$ by sampling either $I$ or $\pm Z$ with equal probability. For each $i$ where the $i^{th}$ entry of $S$ is a $\bullet$, we deterministically set $P_i=I$.
	\item We construct the $n$-qubit Pauli operator ${P}:=\otimes_{i=1}^n {P}_i$, (including its sign $\pm$).
	\item Using the Gottesman-Knill theorem~\cite{aaronson2004improved}, we compute the Pauli operator $P'=\otimes_{i=1}^n P'_i:=U^{\dagger}{P}U$.
	\item We compute the single sample estimate $\hat{p}_1$ using the equation:
			\begin{equation}
				\hat{p}_1:=\tr(\rho P')
				=\prod_{i=1}^n \tr(\rho_i P'_i)\,.
			\end{equation}
	\item We compute the estimator $\hat{p}_s$ by computing $s$ independent single sample estimates and taking their average.
\end{enumerate}

	It is straightforward to show that the expectation value of $\hat{p}_s$ is the target quantum probability $p:=\mc{P}(S)$. Further, the single sample estimates are bounded in the interval $[-1,1]$. Hence, by the Hoeffding inequality,
	\begin{align}
		\pr(|\hat{p}_s-p|\geq \epsilon)\leq 2 e^{\frac{-s\epsilon^2}{2}}\text{.}
	\end{align}
	This algorithm can be executed efficiently in $s$ and in $n$ and produces additive polynomial precision estimates of $\mc{P}(S)$ for any circuit $c \in$ \cmany ~and any $S\in \set{0,1,\bullet}^n$ and is thus a poly-box.
\end{proof}

\subsubsection{A poly-box over \ciqp}

As an additional example, we note that \ciqp, the family of Instantaneous Quantum Polynomial-time (IQP) quantum circuits~\cite{shepherd2009temporally, shepherd2010quantum, bremner2010classical, bremner2016average} that consist  of computational basis preparation and measurements with all gates diagonal in the $X$ basis admits a poly-box.  One can construct such a poly-box over \ciqp~by noting that Proposition 5 from Ref.~\cite{shepherd2010binary} gives a closed form expression for all Born rule probabilities and marginals of these circuits. This expression:
\begin{align}
	\mc{P}_{P}(S)= \mathbb{E}_{r \in {\rm span} \set{\vec{e}_i~|~i\in \set{i_1, \ldots, i_k}}} \left[(-1)^{r \cdot s} \alpha\left(P_{r}, \frac{\pi}{4}\right) \right]
\end{align}
is an expectation value over $2^k$ vectors in $\mathbb{Z}_2^n$ where:
\begin{itemize}
	\item $\set{i_1, \ldots, i_k}$ are the set of indices where the entries of $S$ are in $\set{0,1}$;
	\item $s\in \mathbb{Z}_2^n$ is defined by $s_i=S_i$ when $i\in \set{i_1, \ldots, i_k}$ and $s_i=0$ otherwise;
	\item $P_r$ is the \emph{affinification} of the $m \times n$ binary matrix $P$ which defines a Hamiltonian of the IQP circuit constructed from Pauli $X$ operators according to $H_P:=\sum_{i=1}^m \otimes_{j=1}^n X^{P_{ij}}$;
	\item $\alpha (P, \theta)$ is the normalized version of the weight enumerator polynomial (evaluated at $e^{-2 i \theta}$) of the code generated by the columns of $P$.
\end{itemize}

We note that this is an expectation over exponentially many terms which have their real part bounded in the interval $[-1,1]$. Further, for each $r$, the quantity $\alpha\left(P_{r}, \frac{\pi}{4}\right)$ can be evaluated efficiently using Vertigan's algorithm \cite{vertigan1998bicycle} and Ref.~\cite{shepherd2010binary}. As such, one can construct an additive polynomial precision estimator for all Born rule probabilities and marginals simply by evaluating the expression:
\begin{align}
	\hat{p}_1= {\rm Re} \left[(-1)^{r \cdot s} \alpha\left(P_{r}, \frac{\pi}{4}\right) \right]
\end{align}
 for polynomially many independent uniformly randomly chosen $r \in {\rm span} \set{\vec{e}_i~|~i\in \set{i_1, \ldots, i_k}}$ and computing the average over all choices. This can be shown to produce a poly-box by application of the Hoeffding inequality.

\section{From estimation to simulation}\label{samp}

{{}}{Given the significance of \esim~as the notion that minimally preserves computational power, here we turn our attention to the construction of an efficient algorithms for lifting a poly-box to an \esimer s. We give strong evidence that in the general case, such a construction is not possible. This suggests that a poly-box is statistically distinguishable from an \esimer~and hence computationally less powerful. However, by restricting to a special family of quantum circuits, we show an explicit algorithm for lifting a poly-box to an \esimer. Combined with Theorem \ref{necessity of polybox} this shows that within this restricted family a poly-box is computationally equivalent to an \esimer.}

{{}}{The significance of \esim~also motivates the need to understand the relationship to other simulators defined in terms of Born probability estimation. At the end of this section and in Appendices \ref{strong sim proof} and \ref{mult sim proof} we present two algorithms which lift an estimator of probabilities and marginals to a sampler.}

\subsection{A poly-box is not sufficient for \esim}\label{polybox not suff}

{{}}{This section focuses on the relation between poly-boxes and \esim.} With a poly-box, one can efficiently estimate Born rule probabilities of outcomes of a quantum circuit with additive precision.   {{}}{However, assuming BQP$\neq$ BPP, a poly-box alone is not a sufficient computational resource for \esim.}  We illustrate this using a simple but somewhat contrived example,  {{}}{wherein an encoding into a large number of qubits is used to obscure (from the poly-box) the computational power of sampling.} 

Define a family of quantum circuits $\mc{C}_e$ using a universal quantum computer as an oracle as follows:
\begin{enumerate}
	\item take as input a quantum circuit description $a \in \mc{A}^*$ (this is a description of some quantum circuit with $n$ qubits);
	\item call the oracle to output a sample outcome from this quantum circuit.  Label the first bit of the outcome by $X$;
	\item sample an $n$-bit string $Y \in \set{0,1}^n$ uniformly at random;
	\item output $Z=(X \oplus \mathrm{Par}(Y), Y) \in \bitstring{n+1}$, where $\mathrm{Par}(Y)$ is the parity function on the input bit-string $Y$.
\end{enumerate}

We note that $\mc{C}_e$ cannot admit an \esimer~unless BQP$\subseteq$BPP, since simple classical post processing reduces the \esimer~over $\mc{C}_e$ to an \esimer~over universal quantum circuits restricted to a single qubit measurement. 

We now show that $\mc{C}_e$ admits a poly-box:
\begin{enumerate}
	\item take as input $a\in \mc{A}^*$, $\epsilon, \delta>0$ and $S\in \set{0,1,\bullet}^{n+1}$.   Our poly-box will output probability estimates that are deterministically within $\epsilon$ of the target probabilities and hence we can set $\delta=0$;
	\item if $S$ specifies a marginal probability i.e. $k<n+1$, then the poly-box outputs the estimate $2^{-k}$ (where $k$ is the number of non-marginalized bits in $S$); otherwise,
	\begin{enumerate}
	\item small $\epsilon$ case: if $\epsilon<1/2^n$, explicitly compute the quantum probability $p:=\pr(X=1)$;
	\item large $\epsilon$ case: if $\epsilon\geq 1/2^n$, output the probability $2^{-(n+1)}$ as a guess.
	\end{enumerate}
\end{enumerate}

This algorithm is not only a poly-box over $\mc{C}_e$ but it in fact outputs probability estimates that have exponentially small precision. 

\begin{lemma}
For all $a\in \mc{A}^*$, $\epsilon>0$ and $S\in \set{0, 1, \bullet}^n$, the above poly-box can output estimates within $\epsilon$ additive error of the target probability using $O(poly(n, 1/\epsilon))$ resources. Further, the absolute difference between estimate and target probabilities will be $\leq \min\set{2^{-(n+1)}, \epsilon}$. 
\end{lemma}

\begin{proof}
We note that the resource cost of this algorithm is $O(poly(n,1/\epsilon))$. Since in the case of small $\epsilon$ it is $O(poly(2^n)) \subseteq O(poly(1/\epsilon))$ and in the case of large $\epsilon$ it is $O(n)$.

We now consider the machine's precision by considering the case with no marginalization and the case with marginalization separately. We restrict the below discussion to the large $\epsilon$ case as the estimates are exact in the alternate case.

Let $z=(z_0, \ldots, z_n)\in \bitstring{n+1}$ be fixed and define $z':=(z_1, \ldots, z_n)$. Then, 
\begin{equation}
	\pr(Z=z)=\pr(Z_0=z_0~|~Y=z')Pr(Y=z')=\pr(X=z_0\oplus \mathrm{Par}(z'))2^{-n}\,.
\end{equation}

So for $S=z$ (i.e. no marginalization), we have an error given by $\underset{r\in \set{p,1-p}}{\max} \abs{2^{-(n+1)}-\frac{r}{2^n}}\leq 2^{-(n+1)}$.

For the case where $S_i=\bullet$ (i.e. there is marginalization over the $i^{th}$ bit only and $k=n$), we note that the quantum marginal probability $p(S)$ is given exactly by:
\begin{equation}
	p(S)=\sum_{z_i=0}^1 \pr(Z=z)
	=\sum_{z_i=0}^1 \pr(X=z_0\oplus \mathrm{Par}(z'))2^{-n}
	={p}2^{-n}+(1-p)2^{-n}
	=2^{-k}\,,
\end{equation}
 where $z_j :=S_j$ for $j\neq i$. This {{}}{implies }that for all $k<n+1$, the quantum probability is exactly $2^{-k}$. Thus, in the worst case (no marginalization and $\epsilon\geq 2^{-n}$), the error is $\leq 2^{-(n+1)}$.

\end{proof}

This example clearly demonstrates that the existence of a poly-box for a class of quantum circuits is not sufficient for \esim.  In the following, we highlight the role of the \emph{sparsity} of the output distribution in providing, together with a poly-box, a sufficient condition for \esim.

\subsection{Sparsity and sampling}

Despite the fact that in general the existence of a poly-box for some family $\mc{C}$ does not imply the existence of an \esimer ~for $\mc{C}$, for some quantum circuit families, a poly-box suffices. Here, we show that one can construct an \esimer~for a family of quantum circuits $\mc{C}$ provided that there exists a poly-box over $\mc{C}$ and that the family of probability distributions corresponding to $\mc{C}$ satisfy an additional constraint on the \emph{sparsity} of possible outcomes.
We begin by reviewing several results from Schwarz and Van den Nest~\cite{schwarz2013simulating} regarding sparse distributions.  In Ref.~\cite{schwarz2013simulating}, they define the following property of discrete probability distributions:

\defn{(\ets). A discrete probability distribution is \ts if at most $t$ outcomes have a non-zero probability of occurring. A discrete probability distribution is \ets~if it has a ${L}_1$ distance less than or equal to $\epsilon$ from some \ts distribution.}

The lemma below is a (slightly weakened) restatement of Theorem 11 from Ref.~\cite{schwarz2013simulating}.

\begin{lemma}\label{SchwarzVdNLem}
{(Theorem 11 of Ref.~\cite{schwarz2013simulating}). Let $\mc{P}$ be a distribution on $\bitstring{k}$ that satisfies the following conditions:
\begin{enumerate}
	\item $\mc{P}$ is promised to be \ets, where $\epsilon \leq 1/6$;
	\item For all $S\in \set{0,1,\bullet}^k$, there exists an $(s,k)$-efficient randomized classical algorithm for sampling from $\hat{p}_s$, an additive polynomial estimator of $\mc{P}(S)$.
\end{enumerate}
Then it is possible to classically sample from a probability distribution $\mc{P}' \in B(\mc{P},12\epsilon+\delta)$ efficiently in $k$, $t$, $\epsilon^{-1}$ and $\log\delta^{-1}$.
}
\end{lemma}

We note that for every discrete probability distribution $\mc{P}$, there is some unique minimal function $t(\epsilon)$ such that for all $\epsilon \geq 0$, $\mc{P}$ is $\epsilon$-approximately $t$-sparse. We note that if this function is upper-bounded by a polynomial in $\epsilon^{-1}$, then a randomized classical algorithm for sampling from estimators of $\mc{P}(S)$ can be extended to a randomized classical algorithm for sampling from some probability distribution $\mc{P}'\in B(\mc{P}, \epsilon)$ efficiently in $\epsilon^{-1}$. This fact motivates the following definition:

\defn{(poly-sparse) Let $\mc{P}$ be a discrete probability distribution. We say that $\mc{P}$ is poly-sparse if there exists a polynomial $P(x)$ such that for all $\epsilon > 0$, $\mc{P}$ is \ets~whenever $t\geq P(\frac{1}{\epsilon})$. 

Let $\mathbb{P}$ be a family of probability distributions with $\mc{P}_{a}\in \mathbb{P}$ a distribution over $\bitstring{k_{a}}$.  We say that $\mathbb{P}$ is poly-sparse if there exists a polynomial $P(x)$ such that for all $\epsilon>0$ and $a \in \mc{A}^{*}$, $\mc{P}_{a}$ is \ets~whenever $t\geq P(k_{a}/\epsilon)$.}

The notion of poly-sparse is related to the notion of smooth max entropy $H_{\max}^{\epsilon}$. In particular, $\mathbb{P}$ is poly-sparse iff there exists a polynomial $P(x)$ such that for every $\mc{P} \in \mathbb{P}$ with domain cardinality $2^n$, we have:
\begin{align}
	{{}}{2^{H_{\max}^{\epsilon}(\mc{P})}\leq P\left(\frac{n}{\epsilon}\right)}
\end{align}
where $H_{\max}^{\epsilon}(\mc{P}):= \underset{\mc{P}'}{\inf}~\log_2~|{\rm Supp}(\mc{P}')|$, $|{\rm Supp}(\mc{P}')|$ is the cardinality of the support of the distribution $\mc{P}'$ and the infimum is taken over all distributions $\mc{P}'$ subject to $\frac{1}{2}||\mc{P}'-\mc{P}||_1\leq \epsilon$. This notion was first defined in Ref.~\cite{Renner2004Smooth} where it corresponds to the $\epsilon$-smooth R\'{e}nyi entropy of order $\alpha=0$.

\subsection{Conditions for \esim}

With this notion of output distributions that are poly-sparse, we are in a position to state our main theorem of this section:

\begin{thm}\label{simable}
{Let $\mc{C}$ be a family of quantum circuits with a corresponding family of probability distributions $\mathbb{P}$. Suppose there exists a poly-box over $\mc{C}$, and that $\mathbb{P}$ is poly-sparse. Then, there exists an \esimer~of $\mc{C}$.}
\end{thm}

\begin{proof}
Let $a \in \mc{A}^{*}$ and $\epsilon>0$ be arbitrary. Then there exist $t=t(a,\epsilon)$ such that $\mc{P}_{a}$ is \ets. Further, due to the existence of the efficient classical poly-box over $\mc{C}$, for all $S \in \set{0,1,\bullet}^{k_{a}}$, there exists an $(s, k_{a})$-efficient randomized classical algorithm for sampling from an additive polynomial estimator of $\mc{P}_{a}(S)$. Thus by Lemma \ref{SchwarzVdNLem}, it is possible to classically sample from a probability distribution $\mc{P}_{a}^{\epsilon}\in B(\mc{P}_{a}, \epsilon)$ efficiently in $\epsilon^{-1}, t$ and  $k_{a}$. We note that here we have removed the dependence on $\delta$ since we can make $\delta \leq \epsilon$ whilst remaining efficient in $\epsilon^{-1}, t$ and  $k_{a}$. Finally, since poly-sparsity guarantees the existence of a $t(a, \epsilon)$ that can be upper-bounded by a polynomial in $\frac{k_a}{\epsilon}$, we arrive at the desired result.
\end{proof}

As an example, consider families of quantum circuits $\mc{C}$ where each circuit of size $n$ can only produce outcomes from some set of size at most $poly(n)$. Then $\mc{C}$ is poly-sparse (even if the output distributions are uniform over the $poly(n)$ sized support). Hence, if $\mc{C}$ also admits a poly-box, then by Thm.~\ref{simable} one can with high probability repeatedly sample from this space of $poly(n)$ outcomes hidden within a exponentially large space of bit-strings.

We have shown that having a poly-box and a poly-sparsity guarantee for a family of quantum circuits gives us an \esimer. We emphasize that the proof of this Theorem is constructive, and allows for new simulation results for families of quantum circuits for which it was not previously known if they were efficiently simulable.  As an example, our results can be straightforwardly used to show that Clifford circuits with sparse outcome distributions and with small amounts of local unitary (non-Clifford) noise, as described in Ref.~\cite{bennink2017unbiased}, are \esimable.

{{}}{Theorem \ref{simable} requires a promise of poly-sparsity. Since this is a property of infinite families of probability distributions, one cannot hope to algorithmically verify (or even falsify) it through sampling from member distributions. Nevertheless, for distributions generated by some particular family of quantum circuits, a proof that this property holds may be possible.}

In summary, the results of Thms.~\ref{simable} and \ref{necessity of polybox} imply that in order to construct an \esimer ~of any particular family of quantum circuits, it is necessary to construct a poly-box and further, if the family is poly-sparse, this is also sufficient. In Sec.~\ref{polybox not suff}, we also showed that there exists a somewhat artificial family of quantum circuits $\mc{C}_e$ with respect to which a poly-box is insufficient for \esim. In the next section, we show that this phenomenon also occurs with much more natural families of quantum circuits.

\subsection{{{}}{On lifting stronger estimators to approximate samplers}}

{{}}{
In contrast to poly-boxes, certain stronger nations of simulation based on Born rule probability estimation can be lifted to \esimer s (or even stronger approximate weak simulators). In Appendices  \ref{strong sim proof} and \ref{mult sim proof} we present two such efficient classical algorithms.

The algorithm presented in Appendix \ref{mult sim proof} uses an estimator with multiplicative precision to construct an \esimer~(it can in fact construct an approximate weak simulator based on the stronger notion from Ref.~\cite{Terhal2004}). This algorithm exploits the fact that ratios of multiplicative precision estimators are multiplicative precision in order to sequentially, one qubit's measurement outcome at a time, sample from the marginal probability of the next qubit's measurement conditioned on the sampled measurement outcomes of the prior measurements. This algorithm and its variants have been presented in Refs.~\cite{valiant2002quantum, terhal2002classical, Terhal2004} and are well known within the simulation-of-quantum-circuits community.}

{{}}{The algorithm presented in Appendix \ref{strong sim proof} uses an estimator with exponentially small additive precision to construct an \esimer~(it can in fact construct an approximate weak simulator based on the stronger notion from Ref.~\cite{jozsa2003role}). This algorithm aims to map a bit-string $r$ (approximately representing a uniformly sampled point from the unit interval) to a bit-string representing the outcome of running the circuit. Such a mapping is defined for every ordering of the measurement outcomes. This algorithm makes intuitive use of marginal probability estimates to do a binary search for the measurement outcome corresponding to $r$. This technique avoids computing ratios of probability estimates making it useful in regimes where additive errors are small but larger than some of the probabilities in the target distribution. Hence, this algorithm has some advantages compared to that of Appendix \ref{mult sim proof}. In particular, it can be used to lift an additive $\varepsilon$ precision estimator to a sampler from within $L_1$ distance $O(2^{n}\varepsilon)$. This can be used to construct a \esimer~ in certain cases where the algorithm in Appendix \ref{mult sim proof} would fail. An example is when one has access to an estimator with additive precision $\varepsilon=2^{-n} \kappa$ where $\kappa>0$ can be made arbitrarily small in run-time $O(poly(n,1/\kappa))$.}

\section{Hardness results}\label{hardness}

In the previous section, we have shown that one can construct an \esimer~for a family of quantum circuits $\mathcal{C}$ given a poly-box for this family together with a promise of poly-sparsity of the corresponding probability distribution.  We also discussed a contrived construction of a family of quantum circuits that admits a poly-box but is not \esimable~{{}}{(unless BQP=BPP)}. In this section, we provide strong evidence (dependent only on standard complexity assumptions and a variant of the now somewhat commonly used \cite{aaronson2011computational, bremner2016average, gao2017quantum, bermejo2017architectures, fefferman2015power, morimae2017hardness, bouland2017quantum} ``average case hardness'' conjecture) that {{}}{a condition such as poly-sparsity}  is necessary even for natural families of quantum circuits. One such family has already been identified by noting that \ciqp~admits a poly-box and is likely hard to $\epsilon$-simulate~\cite{bremner2016average}. Here, we also show the likely hardness of $\epsilon$-simulating the non-poly-sparse Clifford circuit family~\cmany~(defined in Sec.~\ref{estim}). {{}}{These results mean that at least two (and possibly more) of the intermediate models of quantum computing have the property that the probability of individual outcomes and marginals can be estimated to $1/poly(n)$ additive error but due to non-sparsity, their $\epsilon$-simulability is implausible.} 

Our hardness result for classical \esim ~of \cmany ~closely follows the structure of several similar results, and in particular that of the IQP circuits result of Ref.~\cite{bremner2016average}. {{}}{We note that this hardness result is implied by the hardness results presented in Refs.~\cite{gao2017quantum, bermejo2017architectures}, however; our proof is able to use a more plausible average case hardness conjecture than these references due to the fact that we are proving hardness of \esim ~rather than proving the hardness of the yet weaker notion of approximate weak simulation employed by these references. }

 Despite the existence of a poly-box over \cmany, we show that there cannot exist a classical \esimer~of this family unless the {{}}{average case hardness conjecture fails or} the polynomial hierarchy collapses to the third level. We note that the hardness of \emph{exact} weak simulation of \cmany~was shown in Ref.~\cite{jozsa2014classical}. In contrast here we show the hardness of \esim~for this family. Our proof relies on a conjecture regarding the hardness of estimating Born rule probabilities to within a small multiplicative factor for a substantial fraction of randomly chosen circuits from \cmany. This average case hardness conjecture is a strengthening of the worst case hardness of multiplicative precision estimation of probabilities associated with circuits from \cmany.  

The hardness of $\epsilon$-simulating \cmany ~circuits is shown by first noting that the existence of a classical \esimer~implies, via the application of the Stockmeyer approximate counting algorithm \cite{stockmeyer1983complexity}, the existence of an algorithm (in the third level of the PH) for estimating the probabilities associated with the output distribution of the \esimer ~to within a multiplicative factor. These estimates can then be related to estimates of the exact quantum probabilities by noting two points:
 \begin{enumerate}
	 \item that the deviation between the \esimer's probability of outputting a particular outcome and that of the exact quantum probability will be exponentially small for the vast majority of outcomes.  We show this fact using Markov's inequality.
	\item that a significant portion of outcomes associated with randomly chosen circuit in \cmany ~must have outcome probabilities larger than a constant fraction of $2^{-n}$.  We show this property using our proof that these circuits anti-concentrate.
 \end{enumerate}
These observations are combined to show that if there exists an \esimer~of \cmany, then there exists a classical algorithm (in the third level of the PH) that can estimate Born rule outcome probabilities to within a multiplicative factor for almost 50\% of circuits sampled from \cmany. This is in contradiction with Conjecture \ref{conj} thus implying that an \esimer~does not exist.

\subsection{Conjecture regarding average case hardness}
\label{subsec:Conjecture}

We begin by stating our conjecture that multiplicative precision estimation of \cmany ~is $\#\text{P}$-hard in the average case.

	\begin{conj}\label{conj}
		There exist an input product state $\rho$ over $n$ qubits such that given a uniformly random Clifford unitary $U$ acting on $n$ qubits, estimating $p:=\tr\left(U\rho U^{\dagger} |0 \rangle \langle 0| \right)$ to within a multiplicative error of $1/poly(n)$ for $49\%$ or more of the sampled Clifford unitaries is $\#$P-hard.
	\end{conj}

We note that this average case hardness conjecture has an analogous worst case hardness version\footnote{This is the same statement as per Conjecture \ref{conj} but with ``$49\%$ or more'' replaced by ``$100\%$''.}. The worst case
hardness can be proven {{}}{by applying the result of Refs.~\cite{Bravyi2005magic, jozsa2014classical, goldberg2017complexity, kuperberg2015hard, Fujii2017commuting} and } by an argument essentially identical to the proof of Theorem 5.1 in Ref.~\cite{bouland2017quantum}. We omit the proof here but note that this proof relies on three key facts:
\begin{enumerate}
	\item that estimating Born rule probabilities for universal (indeed even IQP) circuits{{}}{ that use a gate set with algebraic entries,} to within any multiplicative factor {{}}{in the open interval $(1,\sqrt{2})$}  is $\#$P-hard {{}}{\cite{Fujii2017commuting, goldberg2017complexity}} ;
	\item for gate sets with algebraic entries, all non-zero output probabilities are lower bounded by some inverse exponential~\cite{kuperberg2015hard};
	\item that \cmany ~circuits with post-selection (or adaptivity) are universal for quantum computation {{}}{\cite{Bravyi2005magic, jozsa2014classical}}.
\end{enumerate}

We emphasize that similar conjectures are commonly used in related hardness proofs, such as Refs.~\cite{aaronson2011computational, bremner2016average, gao2017quantum, fefferman2015power, morimae2017hardness, bouland2017quantum}.  

\subsection{Anti-concentration of outcomes for \cmany}

Next, we prove that Clifford circuits chosen uniformly at random from the family \cmany ~satisfy an anti-concentration property.  
\begin{lemma}
	{Let $d$ be a prime. For each $n\in \mathbb{N}$, let $c_n$ be an $n$-qudit Clifford circuit chosen by fixing an arbitrary $n$ qudit input 
	state $\rho$, applying a uniformly random Clifford unitary $U$ acting on $n$ qudits and doing a computational basis measurement on all qudits.  Then for all $\alpha\in (0,1)$ and for any fixed choice of $x \in \set{0,\ldots,d-1}^n$:
	\begin{align}
		\underset{U}{\pr} \left(p_x \geq \frac{\alpha}{d^n} \right)> \frac{(1-\alpha)^2}{2}\text{,}\label{anticoncentration}
	\end{align}
	where $p_x :=\tr\left(U\rho U^{\dagger} |x \rangle \langle x| \right)$ is the Born rule probability for the outcome $x$.
	}
	\end{lemma}\label{anti-concentration}
	
	\begin{proof}
		We use the unitary 2-design property of the Clifford group.
	\begin{align}
		\mathbb{E}(p_x)
		&=
		\tr\left( \mathbb{E}(U \rho U^\dagger) |x\rangle \langle x|\right) =
		\tr\left( \mathbbm{1}/{d^n} |x\rangle \langle x|\right) 
		=\frac{1}{d^n}\label{first moment} \\
		\mathbb{E}(p_x^2)&=
		\tr\left( \mathbb{E}\big(U\otimes U (\rho \otimes \rho) U^\dagger \otimes U^\dagger\big) |x\rangle |x\rangle\langle x| \langle x|\right) \notag\\
		&=
		\frac{\tr\big( P_{\mathrm{Sym}} (\rho \otimes \rho)\big)}{\tr P_{\mathrm{Sym}}}\,
		\tr\left(  P_{\mathrm{Sym}} |x\rangle |x\rangle\langle x| \langle x|\right) \notag\\
		&=
		\frac{2 \tr\big( P_{\mathrm{Sym}}(\rho \otimes \rho)\big)}{d^n(d^n+1)} \notag\\
		&=
		\frac{\left(\tr(\rho^2) + (\tr\rho)^2\right)}{d^n(d^n+1)}\notag\\
		&\leq
		\frac{2}{d^n(d^n+1)}, \label{second moment}
	\end{align}
	where $P_{\mathrm{Sym}}=\frac12(\mathbbm{1} + \mathrm{SWAP})$ is the projection onto the symmetric subspace of $\mathbb{C}^{d^n}\otimes \mathbb{C}^{d^n}$.
We use the Paley-Zygmund inequality, which states that for a non-negative random variable $R$ with finite variance, and for any $\alpha\in (0,1)$:
\begin{align}
	\pr\left(R\geq \alpha{\mathbb{E}[ R]} \right)\geq (1-\alpha)^2 \frac{\mathbb{E}^2[R]}{\mathbb{E}[R^2]}\,,\qquad \text{(Paley-Zygmund inequality)}
\end{align}
Application of this inequality with Eqs.~(\ref{first moment}-\ref{second moment}) then gives the desired result.
	\end{proof}

{{}}{We point out that the property of anti-concentration is inconsistent with poly-sparsity. This result is shown in Theorem \ref{ac vs ps thm} of Appendix \ref{ps vs ac section}.}
	
\subsection{Hardness theorem}
	
We are now in a position to prove our main theorem: 
	\begin{thm}\label{esim no go}
	If there exists an \esimer~of \cmany ~and Conjecture~\ref{conj} holds, then the polynomial hierarchy collapses to the third level.
	\end{thm}
	\begin{proof}
	 Assuming there exists an \esimer~of \cmany, we can treat the \esimer~as a deterministic Turing machine with a random input. Let $\mc{T}$ be the Turing machine that takes as an input $\epsilon>0$ (representing the $L_1$ error required), $r\in \bitstring{poly(n/\epsilon)}$ (representing the random bit-string) and $d_c\in \mc{A}^{poly(n)}$ (representing an efficient description of an $n$ qubit circuit $c\in$ \cmany) and outputs an outcome $X^{\epsilon}\in \bitstring{k}$ with the correct statistics (over uniformly random $r$ inputs) up to $\epsilon$ in $L_1$ distance in time ${\rm poly}(n,1/\epsilon)$. That is, the output satisfies:
\begin{align}
	\dist{p}{p^{\epsilon}}:=\sum_{x\in \bitstring{k}}\abs{p_x-p^{\epsilon}_x}\leq \epsilon
\end{align}
where $p_x:=\pr(X=x)$ is the probability of observing outcome $x$ on a single run of the quantum circuit $c$ and $p^{\epsilon}_x:=\underset{r\sim unif}{\pr}(X^{\epsilon}=x)$ is the probability of observing outcome $x$ on a single run of the Turing machine $\mc{T}$ for a uniformly distributed random $r$ and fixed $\epsilon$, $d_c$ inputs.
	
We now note that the problem of computing the proportion $p^{\epsilon}_x$ of bit-strings $r$ that result in $\mc{T}(\epsilon, r, d_c)=x$ is a problem in $\#$P. Thus, the Stockmeyer algorithm gives us a means of estimating $p^{\epsilon}_x$ to within a multiplicative error in the complexity class $\text{FBPP}^\text{NP}$.

More precisely, there exists an algorithm in $\text{FBPP}^\text{NP}$ which will output an estimate $\tilde{p}^{\epsilon}_x$ such that:
\begin{align}
	\abs{p^{\epsilon}_x-\tilde{p}^{\epsilon}_x}\leq \frac{p^{\epsilon}_x}{poly(n)}
\end{align}

Thus we have that for all $c$ and for all $x$:
\begin{align}
	\abs{p_x-\tilde{p}^{\epsilon}_x}&\leq \abs{p_x-p^{\epsilon}_x}+\abs{p^{\epsilon}_x-\tilde{p}^{\epsilon}_x}\notag\\
	&\leq \abs{p_x-p^{\epsilon}_x}+\frac{p^{\epsilon}_x}{poly(n)}\notag\\
	&\leq \abs{p_x-p^{\epsilon}_x}+\frac{p_x+\abs{p_x-p^{\epsilon}_x}}{poly(n)}\notag\\
	&=\abs{p_x-p^{\epsilon}_x}\left(1+\frac{1}{poly(n)} \right)+\frac{p_x}{poly(n)}\label{additive error UB}
\end{align}

We note that the expectation value of $\abs{p_x-p^{\epsilon}_x}$ over random choice of $x\sim {\rm unif}(\bitstring{k})$ is upper-bounded by $2^{-n} \epsilon$. That is:
\begin{align}
	\underset{x}{\mathbb{E}}[\abs{p_x-p^{\epsilon}_x}]&=\frac{1}{2^{k}}\sum_{x}\abs{p_x-p_x^{\epsilon}}=\frac{1}{2^{k}}\dist{p}{p^{\epsilon}}
	\leq\frac{\epsilon}{2^{k}}\notag
\end{align}

Restricting our attention to circuits in \cmany ~where all of the qubits are measured i.e. $k=n$, we have:
\begin{align}
	\underset{x}{\mathbb{E}}[\abs{p_x-p^{\epsilon}_x}]\leq \frac{\epsilon}{2^n}
\end{align}

We apply Markov's inequality, which states that for $R$ a non-negative random variable and $\gamma>0$:
\begin{align}
	\pr\left(R\geq \frac{\mathbb{E}[ R]}{\gamma} \right)\leq \gamma \,, \qquad \text{(Markov's inequality)}
\end{align}
we have that for all $\beta>0$:
\begin{align}
	\underset{x}{\pr}\left(\abs{p_x-p^{\epsilon}_x}\geq \frac{\underset{x}{\mathbb{E}}[\abs{p_x-p^{\epsilon}_x}]}{\beta} \right)\leq \beta
\end{align}

That is:
\begin{align}
	\underset{x}{\pr}\left(\abs{p_x-p^{\epsilon}_x}< \frac{\epsilon}{\beta 2^n} \right)> (1-\beta)
\end{align}

Applying this to the upper bound in Eq.~\eqref{additive error UB}, we find that for all $\beta>0$:
\begin{align}
	\underset{x}{\pr}\left(\abs{p_x-\tilde{p}^{\epsilon}_x}< \frac{\epsilon}{\beta 2^n}\left(1+\frac{1}{poly(n)} \right) +\frac{p_x}{poly(n)} \right)> (1-\beta)\label{prob of B}
\end{align}

For any fixed choices of $\alpha \in (0,1)$, $\beta,\epsilon>0$, let us define the following events:
\begin{itemize}
	\item Event A: $\frac{p_x}{\alpha}\geq \frac{1}{2^n}$
	\item Event B: $\abs{p_x-\tilde{p}^{\epsilon}_x}< \frac{\epsilon}{\beta 2^n}\left(1+\frac{1}{poly(n)} \right) +\frac{p_x}{poly(n)}$.
\end{itemize}

By Eq.~\eqref{anticoncentration}, we have $\underset{U}{\pr}(A)> \frac{(1-\alpha)^2}{2}$ and by  Eq.~\eqref{prob of B}, we have $\underset{x}{\pr}(B)>(1-\beta)$.   Recall that the intersection bound tells us that ${\pr}(A \cap B)\geq \max\{0, {\pr}(A)+{\pr}(B)-1\}$ for events $A$ and $B$.  Thus, we have $\pr(A \cap B)\geq \frac{(1-\alpha)^2-2\beta}{2}$. This immediately implies the following:
\begin{align}
	\underset{U, x}{\pr}\left(\abs{p_x-\tilde{p}^{\epsilon}_x}< \frac{\epsilon p_x}{\alpha \beta}\left(1+\frac{1}{poly(n)} \right) + \frac{p_x}{poly(n)} \right) > \frac{(1-\alpha)^2-2\beta}{2}\label{prob U,x}
\end{align}

This can be further simplified by incorporating the randomness over $x$ into the uniform randomness over the Clifford unitaries. Specifically, let $y\in \bitstring{n}$ be arbitrarily fixed. Further, let $U_x:=\otimes_{i=1}^n X^{x_i}$. Then, noting that for all $n$ qubit Cliffords $V$: 
\begin{align}
\underset{U,U_x}{\pr}(U_x U=V)&=\underset{U,U_x}{\pr}(U=U_x V)\\
&=\underset{U}{\pr}(U= V)\\
&=\underset{U}{\pr}(U_y U= V)
\end{align}

where probabilities over $U_x$ are chosen uniformly over all $x \in \bitstring{n}$. Applying this to Eq.~\eqref{prob U,x} we find that for all $y\in \bitstring{n}$ and for all $n$ qubit product states $\rho$;
\begin{align}
	\underset{U}{\pr}\left(\abs{p_y-\tilde{p}^{\epsilon}_y}< \frac{\epsilon p_y}{\alpha \beta}\left(1+\frac{1}{poly(n)} \right) + \frac{p_y}{poly(n)} \right) > \frac{(1-\alpha)^2-2\beta}{2}
\end{align}

We recall that for an \esimer , $\epsilon>0$ can be made polynomially small efficiently in run-time and $n$. Thus, as an example, we may assign the following scaling to $\alpha, \beta, \epsilon$:
\begin{align}
	\alpha=\frac{1}{n}\,, \qquad
	\beta=\frac{1}{2n^2}\,, \qquad
	\epsilon=\frac{\alpha \beta}{n}\,.
\end{align}
This argument shows that the existence of an \esimer ~of \cmany ~implies that there exists and algorithm in $\Delta^p_3$ that can for any fixed product states $\rho$ and measurement outcomes $x\in \bitstring{n}$, output an $O(1/n)$ multiplicative precision estimate of $p_x:=\tr(U\rho U^{\dagger}\density{x})$ for almost 50\% of randomly uniformly chosen Clifford unitaries $U$ acting on $n$ qubits. That is:
\begin{align}
	\underset{U}{\pr}\Bigl[\abs{p_x-\tilde{p}^{\epsilon}_x} < p_x O(1/n) \Bigr] > \frac{1}{2}-\frac{1}{n}
\end{align}
By conjecture \ref{conj}, this is $\#$P-hard. This implies that a $\#$P-hard problem is solved in $\text{FBPP}^\text{NP}$. By Toda's theorem \cite{toda1991pp}, this collapses the polynomial hierarchy to its third level.   
	\end{proof}

\section{Discussion}

{{}}{There is a substantial and growing body of results showing the classical ``simulability'' of some quantum computers and the hardness of ``simulability'' of others. We hope that the results presented here will significantly inform the interpretation of this literature in relation to the comparison of the computational power of the relevant quantum computer to the computational power of a universal classical computer. For some family of quantum circuits $\mathcal{C}$, these results typically make statements of the form either:}
\begin{itemize}
	\item {{}}{Simulability: $\mathcal{C}$ can be classically ``simulated'' or}
	\item {{}}{Hardness: $\mathcal{C}$ can be classically ``simulated'' implies some implausible outcome}
\end{itemize}
{{}}{In the case of simulability proofs, our results show that whenever the notion of simulation used is stronger or equivalent to \esim, the useful computational power of $\mathcal{C}$ is contained within classical. Further, if the notion of simulation is a poly-box (a weaker notion then \esim ), this still applies provided that $\mathcal{C}$ is poly-sparse. If $\mathcal{C}$ is not known to be poly-sparse but admits a poly-box then, we can still conclude that without non-trivial classical post-processing, $\mathcal{C}$ is incapable of solving decision problems outside of the complexity class BPP.}

{{}}{In the case of hardness proofs, our results show that whenever the notion of simulation used is weaker or equivalent to \esim, it is plausible that the useful computational power of $\mathcal{C}$ is beyond classical. However, for proofs of hardness based on weaker notions of simulation, it may be possible to alter the proof such that it shows the hardness of \esim~(rather than a yet weaker notion) with the added benefit that now the hardness is more plausible.}

{{}}{Some hardness results show the implausibility of classically simulating $\mathcal{C}$ with respect to a notion of simulation much stronger than \esim. Even if quantum computers can reliably achieve such a notion of simulation, these results cannot be seen as showing the implausibility of the existence of efficient classical devices that can be used as a perfectly good computational substitute to $\mathcal{C}$.}

{{}}{The perspective of efficient indistinguishability gives us a natural avenue to defining the set of all problems solvable by a quantum device. We have seen that the minimal notion of simulation to achieve this is \esim; a significantly weaker notion of simulation than many of the notions used in literature \cite{jozsa2003role, Terhal2004, bremner2010classical, morimae2017power}. 
	Thus, the gap between classical and quantum computational power can be closed not only by the development of more powerful classical simulation algorithms but also by significantly reducing the computational hurdle classical devices must overcome in order to act as efficient substitutes to quantum computers.
	Our results exploit this feature in order to show that any family of quantum circuits that both admits a poly-box and satisfies the poly-sparsity condition can be \esimed.  The existence of multiple known constructions of poly-boxes (see Refs.~\cite{aaronson2004improved, stahlke2014quantum, pashayan2015estimating, bravyi2016improved, bennink2017unbiased}) over restricted families of quantum circuits, and in particular Ref.~\cite{pashayan2015estimating}, demonstrates the significant advantages offered by weakening the minimal requirements on classical simulators from the stronger notions of weak simulation to that of \esim.}
	
	{{}}{For any given family of quantum circuits, poly-sparsity can be trivially guaranteed by upper bounding the number of measured qubits by $\log~n$. However, the condition of poly-sparsity permits significantly more complex probability distribution families (including families with exponentially growing support). Future exploration of how to non-trivially guarantee poly-sparsity offers yet more potential for identifying interesting families of quantum circuits that are \esimable~using the techniques outlined here.}

{{}}{In this paper, we have argued that \esim~minimally captures computational power. However, the term ``minimally'' is with respect to the computational power of the referee, which is unbounded in the setting we considered. This raises the importance of future work aimed at defining the notion of simulation which minimally captures the computational power of a quantum computer with respect to a referee that is computationally bounded to universal quantum computation (or equivalent).
In light of this observation, our work suggests that even requiring a simulator to be capable of solving all sampling problems (as defined in Ref.~\cite{aaronson2014equivalence}) solvable by the quantum device is too strong to be minimal (w.r.t. a universal quantum bounded referee).
Future results in this direction would inform us on precisely how to further weaken the notion of sampling problems and to define a yet weaker complexity class than SampBQP (or more generally Samp$\mathcal{C}$) that (w.r.t. a universal quantum bounded referee) minimally captures computational power.}

{{}}{In an experimental setting where there is a constant lower bound to the noise present in the quantum device, the minimal requirements for efficient indistinguishability become yet weaker. In this setting, it is plausible that for IQP circuits and boson sampling circuits, classical computation can achieve the minimal requirements for efficient indistinguishability w.r.t. a universal quantum bounded referee. This possibility is supported by the existence of classical algorithms for simulating noisy IQP circuits \cite{bremner2017achieving} and noisy boson sampling circuits \cite{oszmaniec2018classical}. In the constant lower bounded noise setting, these algorithms fail to achieve efficient indistinguishability w.r.t. a computationally unbounded referee. However, whether or not they achieve efficient indistinguishability w.r.t. a universal quantum bounded referee remains a question to be resolved.}

 {{}}{Aiming to tighten the separation between simulability and hardness is an important goal toward a deeper understanding of the computational power of quantum verses classical circuits. Specifically, the aim is to move towards a full classification of simulablity by gradually reducing the ``unclassified'' space (of parameters describing a quantum computer that are both outside the range to ensure simulabilty and outside the range to ensure hardness of simulability). By focusing on the tension between anti-concentration and poly-sparsity our work has made modest progress in this direction with potential for further consolidation and progress with respect to this aim.}

{{}}{We have shown the poly-sparsity and anti-concentration properties to be mutually exclusive. If we assume that the polynomial hierarchy does not collapse and restrict to quantum computers that admit a poly-box and average case hardness (some plausible candidates being \ciqp, \cmany, and their poly-sparse restricted counterparts) we see that either poly-sparsity holds ensuring $\epsilon$-simulability or anti-concentration holds ensuring hardness.}

{{}}{For general quantum circuit families, poly-sparsity and anti-concentration are not exhaustive. Future work directed towards finding interesting spaces of quantum computers where the two notions are exhaustive would help to classify more of the yet unclassified computers in Fig.~\ref{fig:FC} (admits a poly-box and not poly-sparse). Restricted to this setting, all quantum computers that admit the appropriate average case hardness property would admit a hardness proof. Further, such work can give a much needed new perspective on the peculiar nature of the transition from \esimable~to hardness that IQP and magic state injected Clifford circuit families undergo as they transition from poly-sparse to non-poly-sparse. In particular, this may shed light on whether this behavior (shared by IQP, magic state injected Clifford circuit and possibly others) is common to intermediate models of quantum computing for a good reason or simply a coincidence.}

{{}}{Our work establishes the conceptual importance of a poly-box as a notion of simulation. Through the Hoeffding inequality and powerful sampling techniques such as Monte Carlo simulations, we inherit a number of important examples of poly-boxes including IQP circuits, magic state injected Clifford circuit and circuits with polynomially bounded negativity (see also Refs.~\cite{aaronson2004improved, stahlke2014quantum, pashayan2015estimating, bravyi2016improved, bennink2017unbiased}). This is of immediate practical interest as admitting a poly-box is sufficient for many useful problems such as finding certain expectation values or estimating the probability associated with certain events.}

{{}}{Whether or not a family of quantum circuits $\mathcal{C}$ admits a poly-box significantly informs our understanding of the computational power of $\mathcal{C}$ relative to classical. Simulability of a family of quantum circuits $\mathcal{C}$ according to the notion of a poly-box, implies that, without additional classical computational resources, $\mathcal{C}$ cannot solve decision problems outside of classical. If $\mathcal{C}$ is also poly-sparse (with binary outcome circuits being a very special case) then even an agent with universal classical computational power and access to the quantum computer $\mathcal{C}$ is confined to universal classical computational power. However, when supplemented with a universal classical computer, if $\mathcal{C}$ admits a poly-box but is not poly-sparse then it may be capable of solving decision problems beyond BPP. This possibility is not ruled out by our analysis and is consistent with the fact that \cmany~and \ciqp~circuits both admit hardness proofs.}

{{}}{There is something conceptually unclear about circuit families that admit poly-boxes and a hardness proof of the type presented in Sec.~\ref{hardness}. In particular, it is unclear if these admit a poly-box purely due to the restriction placed on the types of events that a poly-box can be queried about, or if hardness of \esim~could manifest even in circuit families which allow efficient classical polynomial precision estimation of probabilities associated with any family of events decidable in BPP. In the latter case, an agent with access to such circuits cannot solve any decision problem outside of BPP even given access to a universal classical computer. The former case leaves open the possibility that these families of circuits will behave like \cencode~(introduced in Sec.~\ref{polybox not suff}) where some appropriate classical post-processing of outcome samples will render them more powerful than BPP (assuming BQP$\not \subseteq$ BPP).
	This question is closely related to an open question raised by Aaronson in Ref.~\cite{aaronson2014equivalence}.}
	
	{{}}{It is surprising that examples of families of circuits that admit a poly-box and a plausible hardness proof are far from rare and in fact may be typical among intermediate models of quantum computing. In addition to the families we have shown to be in this category (\cmany~ and \ciqp), we note that linear optical networks \clon~and circuits with polynomially bounded negativity \cpolyneg~are also plausible candidates.
		We note that due to an algorithm by Gurvits~\cite{gurvits2005complexity} (see also Ref.~\cite{aaronson2014generalizing}), the family of linear optical quantum circuits considered in the boson sampling setting of Ref.~\cite{aaronson2011computational} admit additive polynomial precision estimators of individual outcome probabilities. However, there is no known poly-box over this family since it is currently unclear how to produce such estimators for \emph{all} marginal probabilities.
		Alternatively, \cpolyneg~ is known to admit a poly-box \cite{pashayan2015estimating}. Also, for odd prime $d$ it contains the qudit generalization of \cmany~which is both universal under post-selection and anti-concentrates. Hence an average case hardness conjecture is also plausible implying that \cpolyneg~admits a proof of hardness essentially identical to that of \cmany.}
		In light of these considerations we are optimistic that useful and computationally interesting applications can be found for intermediate models of quantum computation.

\begin{acknowledgments}
The authors are grateful to Marco Tomamichel, Ryan Mann, Michael Bremner, Dax Koh, Daniel Brod and Joel Wallman for helpful discussions. This research is supported by the ARC via the Centre of Excellence in Engineered Quantum Systems (EQUS) project number CE170100009 and by the U.S. Army Research Office through grant W911NF-14-1-0103. HP also acknowledges support from the Australian Institute for Nanoscale Science and Technology Postgraduate Scholarship (John Makepeace Bennett Gift).  The work of DG is supported by the Excellence Initiative of the German Federal and State Governments (ZUK 81), the DFG within the CRC 183 (project B01), and the DAAD.

\end{acknowledgments}

\newpage

\appendix

\section{Statistical indistinguishability proof}\label{indistinguishability proof}

We first show a well known connection between the optimal probability of choosing the correct hypothesis in a hypothesis test and the $L_1$ distance.

Suppose $\mc{P}_1$ and $\mc{P}_2$ are probability distribution over some finite set $I$, and suppose a sample $X$ is observed from the distribution $\mc{Q}$ where either $\mc{Q}=\mc{P}_1$ (hypothesis $H_1$) or $\mc{Q}=\mc{P}_2$ (hypothesis $H_2$).  Then, any hypothesis test must have some $H_1$ acceptance region $A_1\subseteq I$ and some $H_2$ acceptance region $A_2:=A_1^c\subseteq I$. The probability of a type I error is $\alpha:= \pr(X\in A_2 ~|~ X\sim \mc{P}_1)$ and the probability of a type II error is $\beta:= \pr(X\in A_1 ~|~ X\sim \mc{P}_2)$. The $L_1$ distance between $\mc{P}_1$ and $\mc{P}_2$ can be written as:
\begin{align*}
	\dist{\mc{P}_1}{\mc{P}_2}:&=\sum_{x\in I} |\mc{P}_1(x)-\mc{P}_2(x)|\\
	&=2 \underset{A_1\subset I}{\sup}~[\mc{P}_1(A_1)-\mc{P}_2(A_1)]\\
	&=2 \underset{A_1\subset I}{\sup}~[(1-\mc{P}_1(A_1^c) - \mc{P}_2(A_1)]\\
	&=2(1-\alpha^*-\beta^*)
\end{align*}
where, the second equality can be verified by noting that the supremum is achieved when $A_1=\set{x\in I~|~\mc{P}_1(x)\geq \mc{P}_2(x)}$. Here, $\alpha^*$ and $\beta^*$ are the type I and type II errors for the optimal choice of acceptance region / hypothesis test. We note that if a priori, $H_1$ and $H_2$ are equally likely, then the probability of choosing the correct hypothesis, based on a single sample, using the optimal test is thus given by:
\begin{align}
	P_{correct}&=1-Pr(X\in A_2 ~|~ X\sim \mc{P}_1) \pr(X\sim \mc{P}_1)-Pr(X\in A_1 ~|~ X\sim \mc{P}_2) \pr(X\sim \mc{P}_2)\notag\\
	&=1-\alpha^* \pr(H_1)-\beta^* \pr(H_2)\notag\\
	&=\frac{1}{2}+\frac{\dist{\mc{P}_1}{\mc{P}_2}}{4}\label{Pcorrect bound}\,.
\end{align}

\begin{figure}[h!]
  \centering
	\includegraphics{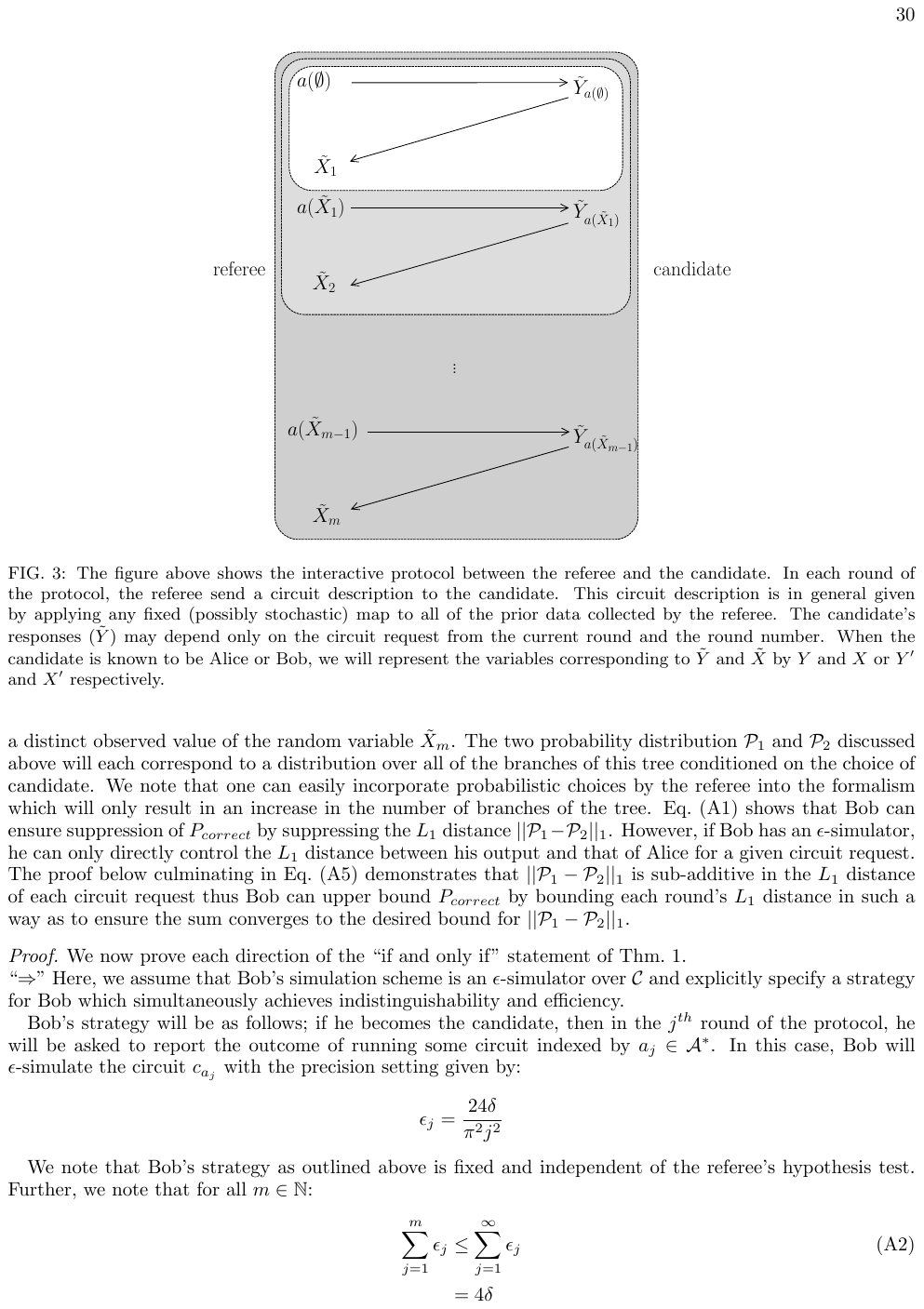}
  \caption{\label{interactive protocol} The figure above shows the interactive protocol between the referee and the candidate. In each round of the protocol, the referee send a circuit description to the candidate. This circuit description is in general given by applying any fixed (possibly stochastic) map to all of the prior data collected by the referee. The {{}}{candidate's} responses ($\tilde{Y}$) may depend only on the circuit request from the current round and the round number. When the candidate is known to be Alice or Bob, we will represent the variables corresponding to $\tilde{Y}$ and $\tilde{X}$ by ${Y}$ and ${X}$ or ${Y'}$ and ${X'}$ respectively.}
  \end{figure}

The interactive protocol between the referee and the candidate will proceed as follows (see Figure \ref{interactive protocol}):
\begin{enumerate}
	\item Initially, the referee will fix a test by choosing a function  {{}}{$a(\cdot)$} that dictates how all gathered data in prior rounds determines the next circuit request.  
	We note that while this can be further generalized by allowing stochastic maps (rather than functions), this has no baring on our results and our proof can fairly easily be extended if required.
	\item Initially the referee will make the circuit request $a_{\emptyset}\in \mc{A}^*$
	\item The response from the candidate is denoted by the random variable $\tilde{Y}_{a_{\emptyset}}$ and the string of random variables $a_{\emptyset}, \tilde{Y}_{a_{\emptyset}}$ will be represented by $\tilde{X}_1$
	\item The referee may make another circuit request by applying the map $a$ to $\tilde{X}_1$ thus defining the next circuit request $a(\tilde{X}_1)$.  
	\item On the $(j+1)^{th}$ round, the referee's circuit request will be represented by $a(\tilde{X}_j)$ and the response will be represented by $\tilde{Y}_{a(\tilde{X}_j)}$ where, $\tilde{X}_{j+1}$ represents the string of random variables $\tilde{X}_{j}, a(\tilde{X}_j), \tilde{Y}_{a(\tilde{X}_j)}$.
	\item In addition, at the end of the $j^{th}$ round for $j=1, 2, \ldots$, a fixed stochastic binary map $h$ will be applied to $\tilde{X}_j$ with the outcome determining whether or non to halt the interactive procedure. We will assume that the test will eventually halt and represent the final round of any given test by $m\in \bbn$. 
	\item Finally, the referee will decide $H_a$ vs $H_b$ by applying a fixed binary map $d$ to the full collected data set $\tilde{X}_m$. 
\end{enumerate}

We will use the notation convention above but in the case when the candidate is fixed to be Alice, we will remove the tilde (i.e. $\tilde{X}, \tilde{Y} \rightarrow X, Y$) and alternatively when the candidate is fixed to be Bob, we will replace the tilde with a prime (i.e. $\tilde{X}, \tilde{Y} \rightarrow X', Y'$).

{{}}{The set of all possible data collected by the referee (based on all probabilistic choices including the choice of the candidate) over the course of the entire test can be viewed as a tree where each branch corresponds to a distinct observed value of the random variable $\tilde{X}_m$. The two probability distribution $\mc{P}_1$ and $\mc{P}_2$ discussed above will each correspond to a distribution over all of the branches of this tree conditioned on the choice of candidate. We note that one can easily incorporate probabilistic choices by the referee into the formalism which will only result in an increase in the number of branches of the tree. Eq.~\eqref{Pcorrect bound} shows that Bob can ensure suppression of $P_{correct}$ by suppressing the $L_1$ distance $\dist{\mc{P}_1}{\mc{P}_2}$. However, if Bob has an \esimer, he can only directly control the $L_1$ distance between his output and that of Alice for a given circuit request. The proof below culminating in Eq.~\eqref{subadditivity} demonstrates that $\dist{\mc{P}_1}{\mc{P}_2}$ is sub-additive in the $L_1$ distance of each circuit request thus Bob can upper bound $P_{correct}$ by bounding each round's $L_1$ distance in such a way as to ensure the sum converges to the desired bound for $\dist{\mc{P}_1}{\mc{P}_2}$.
}

\begin{proof}
We now prove each direction of the ``if and only if'' statement of Thm.~\ref{indistinguishability thm}.
\newline ``$\Rightarrow$''
Here, we assume that Bob's simulation scheme is an \esimer ~over $\mc{C}$ and explicitly specify a strategy for Bob which simultaneously achieves indistinguishability and efficiency. 

Bob's strategy will be as follows; if he becomes the candidate, then in the $j^{th}$ round of the protocol, he will be asked to report the outcome of running some circuit indexed by $a_j\in \mc{A}^*$. In this case, Bob will $\epsilon$-simulate the circuit $c_{a_j}$ with the precision setting given by:
\begin{align*}
	\epsilon_j=\frac{24\delta}{\pi^2 j^2}
\end{align*}

We note that Bob's strategy as outlined above is fixed and independent of the referee's hypothesis test. Further, we note that for all $m\in \bbn$:
\begin{align}
	\sum_{j=1}^m \epsilon_j &\leq \sum_{j=1}^{\infty} \epsilon_j \label{Bob error}\\
	&=4\delta \notag
\end{align}

We define the map $\mc{E}[X,X']$ from any pair of random variables $X$ with probability distribution $\mc{P}$ and $X'$ with probability distribution $\mc{P}'$ to $\bbr$ as the $L_1$ distance between $\mc{P}$ and $\mc{P}'$.
 
We will show that for every test, the quantity on the LHS of Eq.~\eqref{Bob error} upper bounds $\mc{E}[X_m, X'_m]$. Hence:
\begin{align}
	\mc{E}[X_m, X'_m]\leq 4\delta
\end{align}

\begin{align}
\mc{E}[X_{j+1},X'_{j+1}]&=\sum_{\alpha,\beta} \abs{\Pr(Y_{a(X_j)}=\beta|X_j=\alpha)\Pr(X_j=\alpha)-\Pr(Y'_{a(X'_j)}=\beta|X'_j=\alpha)\Pr(X'_j=\alpha)}\notag\\
&=\sum_{\alpha,\beta} |\Pr(Y_{a(X_j)}=\beta|X_j=\alpha)\Pr(X_j=\alpha)-\Pr(Y_{a(X_j)}=\beta|X_j=\alpha)\Pr(X'_j=\alpha)\notag\\
&~~~~~~~~~~~~~~~+\Pr(Y_{a(X_j)}=\beta|X_j=\alpha)\Pr(X'_j=\alpha)-\Pr(Y'_{a(X'_j)}=\beta|X'_j=\alpha)\Pr(X'_j=\alpha)|\notag\\
&\leq \sum_{\alpha,\beta} \Pr(Y_{a(X_j)}=\beta|X_j=\alpha)\abs{\Pr(X=\alpha)-Pr(X'=\alpha)}\notag\\
&~~~~~~~~~~~~~~~+\sum_{\alpha,\beta} \Pr(X'_j=\alpha) \abs{\Pr(Y_{a(X_j)}=\beta|X_j=\alpha)-\Pr(Y'_{a(X'_j)}=\beta|X'_j=\alpha)}\notag\\
&\leq \mc{E}[X_j,X'_j]+\sum_\alpha {\mc{E}[Y_{a(\alpha)},Y'_{a(\alpha)}]\Pr(X'_j=\alpha)}\label{test indep}
\end{align}

where the sums are taken over $\alpha$ in the support of $\tilde{X}_j$ and $\beta$ in $\underset{a\in \mc{A}^*}{\cup}supp(\tilde{Y}_a)$.

We note that the precision of Bob's response in any round only depends on the round number. Thus, Eq.~\eqref{test indep} can be simplified by {{}}{replacing $\mc{E}[Y_{a(\alpha)},Y'_{a(\alpha)}]$ with the upper bound $\epsilon_{j+1}$}. Combined with the observation that $\mc{E}[X_1, X'_1]=\epsilon_1$, we have shown that:

\begin{align}
\mc{E}[X_{m},X'_{m}]&\leq \sum_{j=1}^m \epsilon_j \label{subadditivity}\\
	&\leq 4\delta\notag
\end{align}

This proves that Bob's strategy meets the indistinguishibility property. We now consider the efficiency of the strategy. We recall that given a circuit request sequence $\alpha$, Alice's and Bob's resource costs are represented by $N(\alpha)$ and $T(\alpha)$ respectively. Further, Alice's resource costs is lower bounded by $m$, the number of rounds of the Hypothesis test $\alpha$.

By definition of \esim, there exists $\kappa, c_1, c_2 \in \bbn$ such that for a given circuit index $a$, and precision $\epsilon$, $T(a)\leq c_1\left(\frac{N(a)}{\epsilon}\right)^\kappa +c_2$. For simplicity, we will set $c_1=1$ and $c_2=0$ as this is immaterial given sufficiently large $N(\alpha)$ and $\frac{1}{\epsilon}$. For $m=1$, clearly the strategy is efficient. Hence, given a string of inputs $\alpha=(a_1, \ldots, a_m)$, with $m\geq 2$ we have:
\begin{align}
	T(\alpha)&=\sum_{j=1}^m T(a_j)\\
	&\leq \sum_{j=1}^m \left(\frac{N(a_j)}{\epsilon_j}\right)^\kappa\\
	&=\sum_{j=1}^m \left(\frac{\pi^2 j^2 N(a_j)}{24\delta}\right)^\kappa\\
	&\leq \left(\frac{\pi^2}{24\delta}\right)^\kappa \left[\sum_{j=1}^{m-1} j^{2\kappa} + m^{2\kappa} [N(\alpha)-(m-1)]^\kappa \right] \label{min n is one}\\
	&\leq \left(\frac{\pi^2}{24\delta}\right)^\kappa \left[\left(\frac{m-0.5}{2\kappa +1}\right)^{2\kappa+1}+m^{2\kappa}N(\alpha)^{\kappa}-m^{2\kappa}(m-1)^{\kappa}\right]\label{integ approx}\\
	&\leq\left(\frac{\pi^2 m^2 N(\alpha)}{24\delta}\right)^\kappa \label{m geq 2}\\
	&\leq \left(\frac{\pi^2N(\alpha)^{3}}{24\delta}\right)^\kappa\label{Bob run-time UB}\\
	&\in O\left(poly(N(\alpha), \frac{1}{\delta})\right)
\end{align}
where:
\begin{itemize}
\item in Eq.~\eqref{min n is one} we have used the fact that $N(a_j)\geq 1$ for all $j$ hence the expression is maximized when $\alpha$ is chosen such that $N(a_j)=1$ for $j=1, \ldots, m-1$ and $N(a_m)=N(a_1)+\ldots+N(a_{m})-(m-1)$;
	\item in Eq.~\eqref{integ approx} we have used integration to show the inequality for any $k\in \bbn$; $\sum_{j=1}^m j^k< (\frac{m+0.5}{k+1})^{k+1}$; and 
	\item in Eq.~\eqref{integ approx} we have also used the fact that for $\kappa>1$ and $N\geq m>0$, one can show that $(N-m)^{\kappa}\leq N^{\kappa}-m^{\kappa}$
	\item in Eq.~\eqref{m geq 2} we have used the inequality $\left(\frac{m-0.5}{2\kappa +1}\right)^{2\kappa+1}-m^{2\kappa}(m-1)^{\kappa}\leq 0$ for $m\geq 2$ and $\kappa \geq 1$;
\end{itemize}

hence, there exists a polynomial $f(x,y)$ such that for all request strings $\alpha$ and $\delta>0$, $T(\alpha)\leq f(N(\alpha),\frac{1}{\delta})$.

\vspace{0.5cm}

``$\Leftarrow$'': We restrict ourselves to interactive protocols consisting of only one round. For each fixed circuit request, under the optimal choice of the decision map $d$, $\delta \propto \epsilon$ hence for all $c\in \mc{C}$ and for all $\epsilon>0$, Bob must be able to sample from some distribution $\mc{P}^{\epsilon}\in B(\mc{P},\epsilon)$. Further, since Bob's strategy meets the efficiency condition, for every $a \in \mc{A}^*$, Bob must be able to output the sample using resources $\in O\left(poly(N(a),\frac{1}{\delta})\right)\subseteq O\left(poly(n,\frac{1}{\epsilon})\right)$.
\end{proof}

\section{Strong simulation implies EPSILON-simulation}\label{strong sim proof}

In this appendix , we will show that the existence of a classical strong simulator of a family of quantum circuits implies the existence of an \esimer~ {{}}{(it can in fact construct an approximate weak simulator based on the stronger notion from Ref.~\cite{jozsa2003role})}. {{}}{This algorithm aims to map a bit-string (representing the outcome of running the circuit) to $r$, which is sampled uniformly from [0,1]. While such a mapping is defined for every ordering of the measurement outcomes, it cannot be efficiently computed. This algorithm makes intuitive use of marginal probability estimates to do a binary search for the bit-string corresponding to $r$. This technique avoids computing ratios of probability estimates making it useful in regimes where additive errors are small but larger than some of the probabilities in the target distribution.}

We start by giving a more precise definition of a strong simulator (than was presented in Sec.~\ref{defn strong and weak sim}).

\defn{\emph{(strong simulator).}  A strong simulator of a uniform family of quantum circuits $\mc{C}=\set{c_a~|~a\in \mc{A}^*}$ with associated family of probability distributions $\mathbb{P}=\set{\mc{P}_a~|~a\in \mc{A}^*}$ is a classical algorithm that, for all $a\in \mc{A}^*, \epsilon, \delta>0$ and $S\in \set{0,1,\bullet}^{k_n}$, can be used to compute an estimate $\hat{p}$ of $\mc{P}_{a}(S)$ such that $\hat{p}$ satisfies the accuracy requirement:
	\begin{align}
		\Pr\bigl(\abs{p-\hat{p}}\geq \epsilon\bigr)\leq \delta\label{accuracy req EB}
	\end{align}
	and, the run-time required to compute the estimate $\hat{p}$ is $O(poly(n, \log~\epsilon^{-1}, \log~\delta^{-1}))$.
}

We point out that much like a poly-box, a strong simulator outputs estimates of Born probabilities. The key difference is that the precision of a strong simulator is exponential compared to the polynomial precision of a poly-box. In particular, for any polynomial $f$, a strong simulator can (efficiently in $n$) output estimates such that Eq.~\eqref{accuracy req EB} is satisfied for $\epsilon \in \Omega(2^{-f(n)})$ (as opposed to a poly-box which generally requires $\epsilon \in \Omega(1/{f(n)})$). Hence, we note that the only difference between the definition of a strong simulator and that of a poly-box is the scaling of run-time in $\epsilon$. 

\begin{thm}\label{strong sim thm}
Let $\mc{C}$ be a uniform family of quantum circuits. If $\mc{C}$ admits a strong simulator, then $\mc{C}$ admits an \esimer.
\end{thm}

{{}}{In fact we will prove an even stronger statement; that a strong simulator implies approximate weak simulation in the much stronger sense of approximate weak simulation used in Ref.~\cite{jozsa2003role} (exponentially small error in $L_1$ norm).}

Before proving this theorem, we introduce an algorithm that uses output from a strong simulator to approximately sample from the output distribution of a circuit i.e. to produce output consistent with the definition of an \esimer. Without loss of generality, let $c \in \mc{C}$ be an arbitrary $n$ qubit circuit with all $n$ qubits measured. We will denote the quantum probabilities by $p_S$ and the output of the strong simulator by ${p}^{\epsilon, \delta}_S$ suppressing the dependence on $c$.

To give a rough intuition, the algorithm will first sample a polynomial length bit-string $\tilde{r}$ which will be mapped to a probability $r\in [0,1]$. This value will remain fixed and be used throughout the algorithm until a sample $\tilde{X}$ is generated from the approximate output distribution. This sample will be the output of the \esimer ~upon a single execution with the input $(\epsilon', c)$. The sample $\tilde{X}=(\tilde{X}_1,\ldots, \tilde{X}_n)$ will be generated by sampling one bit at a time starting with $\tilde{X}_1$. The choice of the $j^{th}$ bit $\tilde{X}_j$ is based on the comparisons between the output of the strong simulator ${p}^{\epsilon, \delta}_S$ and the probability $r$. This $n$ step process will require $n$ calls to the strong simulator where in each call, the only variation in the inputs is the events $S_j$. Each event $S_j$ will be chosen based on the previously sampled values $\tilde{X}_1,\ldots, \tilde{X}_{j-1}$.

The algorithm will proceed as follows:
\begin{enumerate}
\item Fix $m\in \bbn$ and $\epsilon, \delta>0$ based on $\mc{C}$ and the desired $L_1$ error upper bound, $\epsilon'$ (see later). 
\item Sample $\tilde{r}$ uniformly from $\bitstring{m}$. 
\item Compute $r=\sum_{i=1}^{m} {\tilde{r}_i}2^{-i}$
\item Set $S:=(s_1, \ldots, s_n)=(\bullet, \ldots, \bullet)$.
\item Set $j=1$.
\item Set $s_j=0$.
\item Set $S_j=S$.
\item Request ${p}^{\epsilon, \delta}_{S_j}$ from the strong simulator.
\item If ${p}^{\epsilon, \delta}_{S_j}\geq r$, then set $\tilde{X}_j=0$ otherwise, set $\tilde{X}_j=1$.
\item Set $s_j=\tilde{X}_j$.
\item If $j=n$, output the string $\tilde{X}=(\tilde{X}_1,\ldots,\tilde{X}_n)$ and end.
\item Reset $j\rightarrow j+1$ and go to step 6.
\end{enumerate}

We now prove Theorem \ref{strong sim thm}.
\begin{proof}
We wish to show that for all acceptable families of quantum circuits $\mc{C}$, choices of $c\in \mc{C}$ and $\epsilon'>0$:
\begin{itemize}
	\item there exist a polynomially bounded function $f(\epsilon',n)$ which determines $m$ and
	\item there exist functions for determining $\epsilon, \delta$
\end{itemize}
 such that given a strong simulator of $\mc{C}$, the above algorithm can be executed in run-time $O(poly(n, \epsilon'^{-1}))$ and produce output $\tilde{X}$ from a distribution $\tilde{\mc{P}}$ satisfying $\tilde{\mc{P}}\in B(\mc{P},\epsilon')$.

We note that the probability distribution over $x\in \bitstring{n}$ defines a partitioning (up to sets of measure zero) of the unit interval into $2^n$ intervals\footnote{Here, we use a looser notion of interval by allowing points $p\in \bbr$ to constitute an intervals $[p,p]$.} $V_x$ labeled by $x$ such that the uniform measure on these intervals corresponds to the quantum probability of outcome $x$. That is, we fix the partitioning such that for all $x\in \bitstring{n}$:
\begin{align}
	\mu(V_x)=p_x\text{.}
\end{align}
To be specific, we can define $V_x=[v^-_x,v^+_x]$ where:
\begin{align}
	v^-_x=\sum_{x'<x} p_{x'}\\
	v^+_x=\sum_{x'\leq x} p_{x'}\\
\end{align}
where, the above order on bit strings $x'$ and $x$ is defined by lexicographical ordering.

We note that given a uniform sample $p$ from the unit interval, $p$ will, up to measure zero, be strictly identified with an outcome $x\in \bitstring{n}$ through the mapping $o: [0,1]\setminus D \rightarrow \bitstring{n}$ implicitly defined by $p\in V_{o(p)}$ for all $p\in [0,1]\setminus D$ where $D:=\set{p_S|S\in \set{0,1,\bullet}^n}$. Further, in the ideal case where the strong simulator produces output which is deterministically exact i.e. ${p}^{\epsilon, \delta}_{S}=p_S$ for all $S$, we note that the above algorithm would, for a given $r$, produce output $\tilde{X}=o(r)$. For $r$ distributed uniformly on the unit interval, this ensures $\tilde{X}$ is sampled from exactly the quantum distribution. We thus note that two sources of error arise. The first is from the inaccuracies introduced by the strong simulator's output. The second is from having to approximate a uniform sample over $[0,1]$ by a uniform sample over $\bitstring{m}$.

Let $\tilde{p}^{\epsilon, \delta}_{x}$ denote the probability $\pr(\tilde{X}=x)$. Then, we have:
\begin{align}
	\tilde{p}^{\epsilon, \delta}_{x}=\sum_{\tilde{r}\in \bitstring{m}} 2^{-m} \pr(\tilde{X}=x~|~r)
\end{align}

Given an interval $V=[v^-,v^+]$ and $\alpha\in \bbr$, we define:
\begin{align}
V^{\alpha}=	\begin{cases}
															[v^- -\alpha, v^+ +\alpha],& \text{if } \alpha\geq 0 \text{ or } v^+-v^-\geq 2\alpha\\
															[\frac{v^- +v^+}{2}, \frac{v^- +v^+}{2}],& \text{otherwise}
													\end{cases}
\end{align}

If $p_x\geq 2\epsilon$ and $r\in V^{-\epsilon}_x$ then:
\begin{align}
	\pr(\tilde{X}=x|r)\geq(1-\delta)^n\text{.}
\end{align}
This can be seen by noting that with probability $\geq 1-\delta$, each requested probability estimate in step 8 will be within $\epsilon$ of the corresponding quantum probability resulting in $\tilde{X}_j=o(r)_j$.

Thus, we have:
\begin{align}
	\tilde{p}^{\epsilon, \delta}_{x}&\geq \underset{r\in V^{-\epsilon}_x}{\sum_{\tilde{r}\in \bitstring{m}}} 2^{-m} (1-\delta)^n\\
	&\geq l_x 2^{-m}(1-\delta)^n\\
	&\geq [p_x-2\epsilon -2^{-m}](1-\delta)^n
\end{align}
where
\begin{align}
l_x:&=\left\lfloor{\frac{\abs{V^{-\epsilon}_x}}{2^{-m}}}\right\rfloor\\
&=\left\lfloor{\frac{p_x-2\epsilon}{2^{-m}}}\right\rfloor
\end{align}
is a lower bound on the number of bit strings $\tilde{r}$ which under the map in step 3 must be contained in the interval $V^{-\epsilon}_x$.

If $r\in \bbr \setminus V^{+\epsilon}_x$ then:
\begin{align}
	\pr(\tilde{X}=x|r)\leq 1-(1-\delta)^n
\end{align}
since estimates within $\epsilon$ of the target probability at each of the $n$ iterations of step 8 will result in $\tilde{X}\neq x$.

Thus, we also have:
\begin{align}
	\tilde{p}^{\epsilon, \delta}_{x}&\leq \sum_{\underset{r\in V^{+\epsilon}_x}{{\tilde{r}\in \bitstring{m}}}} 2^{-m} \pr(\tilde{X}=x~|~r) + \sum_{\underset{r\in \bbr \setminus V^{+\epsilon}_x}{{\tilde{r}\in \bitstring{m}}}} 2^{-m} 1-(1-\delta)^n\\
	&\leq u_x 2^{-m} + \left[1-(1-\delta)^n\right]\\
	&\leq \left[p_x+2\epsilon + 2^{-m}\right] + \left[1-(1-\delta)^n\right]
\end{align}
where
\begin{align}
u_x:&=\left\lfloor{\frac{\abs{V^{+\epsilon}_x}}{2^{-m}}}\right\rfloor+1\\
&=\left\lfloor{\frac{p_x+2\epsilon}{2^{-m}}}\right\rfloor+1
\end{align}
is an upper bound on the number of bit strings $\tilde{r}$ which under the map in step 3 must be contained in the interval $V^{+\epsilon}_x$.

Thus $p_x\geq 2\epsilon$:

\begin{align}
	[-2\epsilon -2^{-m}](1-\delta)^n -p_x \left[1-(1-\delta)^n\right] \leq \tilde{p}^{\epsilon, \delta}_{x}-p_x \leq \left[2\epsilon + 2^{-m}\right] + \left[1-(1-\delta)^n\right]
\end{align}
i.e.
\begin{align}
	\abs{\tilde{p}^{\epsilon, \delta}_{x}-p_x} \leq \left[2\epsilon + 2^{-m}\right] + \left[1-(1-\delta)^n\right]\label{abs prob error strong sim}
\end{align}

Also, if $p_x\leq 2\epsilon$, then:
\begin{align}
	-p_x \leq \tilde{p}^{\epsilon, \delta}_{x}-p_x \leq \left[2\epsilon + 2^{-m}\right] + \left[1-(1-\delta)^n\right]
\end{align}
thus, the bound from Eq.~\eqref{abs prob error strong sim} also applies in this case.

This implies that:
\begin{align}
	\epsilon'\leq 2^n \left[2\epsilon + 2^{-m} + 1-(1-\delta)^n\right]\text{.}\label{strong sim l1 error}
\end{align}
 Clearly there exist choices of polynomials $f_1,f_2,f_3$ such that for $\epsilon'=\frac{1}{poly(n)}$ or even $\epsilon'=2^{-poly(n)}$, Eq.~\eqref{strong sim l1 error} can be satisfied by choosing $\epsilon\leq 2^{-f_1(n)}$, $\delta\leq 2^{-f_2(n)}$ and $m\geq f_3(n)$. We complete the proof by noting that these choices ensure that the run-time of the strong simulator and the above algorithm are efficient in $n$ and $1/\epsilon'$. 
\end{proof}

\section{Multiplicative precision simulation implies EPSILON-simulation}\label{mult sim proof}

{{}}{In this section, we present an algorithm (very similar to Ref.~\cite{Terhal2004}) which uses an estimator with multiplicative precision to construct an \esimer~(it can in fact construct an approximate weak simulator based on the stronger notion from Ref.~\cite{Terhal2004}). This algorithm exploits the fact that ratios of multiplicative precision estimators are multiplicative precision in order to sequentially, one qubit's measurement outcome at a time, sample from the marginal probability of the next qubit's measurement conditioned on the sampled outcomes of the prior measurements. This algorithm and its variants have also been presented in \cite{valiant2002quantum, terhal2002classical} and are well known within the simulation-of-quantum-circuits community.}

{{}}{Here, we claim without proof that this algorithm lifts a classical multiplicative precision simulator of a family of quantum circuits to an approximate weak simulator based on the stronger notion from Ref.~\cite{Terhal2004}.}
 {{}}{This result has been shown in Ref.~\cite{Terhal2004}, but we discuss it here for completeness.} 

We start by giving a definition of a multiplicative precision simulator.

\defn{\emph{(multiplicative precision simulator).}  A multiplicative precision simulator of a uniform family of quantum circuits $\mc{C}=\set{c_a~|~a\in \mc{A}^*}$ with associated family of probability distributions $\mathbb{P}=\set{\mc{P}_a~|~a\in \mc{A}^*}$ is a classical algorithm that, for all $a\in \mc{A}^*, \epsilon, \delta>0$ and $S\in \set{0,1,\bullet}^{k_n}$, can be used to compute an estimate $\hat{p}$ of $\mc{P}_{a}(S)$ such that $\hat{p}$ satisfies the accuracy requirement:
	\begin{align}
		\pr\bigl(\abs{p-\hat{p}}\geq \epsilon p \bigr)\leq \delta\label{accuracy req MPS}
	\end{align}
	and, the run-time required to compute the estimate $\hat{p}$ is $O(poly(n, \epsilon^{-1}, \delta^{-1}))$.
}

We claim that a multiplicative precision simulator can be used to construct an \esimer.

\begin{thm}\label{mult sim thm}
Let $\mc{C}$ be a uniform family of quantum circuits. If $\mc{C}$ admits a multiplicative precision simulator, then $\mc{C}$ admits an \esimer.
\end{thm}

We omit a complete proof of this theorem as it makes straightforward use of standard techniques. However, we outline the algorithms which uses output {{}}{from}  a multiplicative precision simulator to approximately sample from the output distribution of a circuit. Without loss of generality, let $c \in \mc{C}$ be an arbitrary $n$ qubit circuit with all $n$ qubits measured. We will denote the quantum probabilities by $p_S$ and the output of the multiplicative precision simulator by ${p}^{\epsilon, \delta}_S$ suppressing the dependence on $c$.

The algorithm will proceed as follows:
\begin{enumerate}
\item Fix $m\in \bbn$ and $\epsilon, \delta>0$ based on $\mc{C}$ and the desired $L_1$ error upper bound, $\epsilon'$ {(see later)}. 
\item Set $S:=(s_1, \ldots, s_n)=(\bullet, \ldots, \bullet)$.
\item Set ${p}^{\epsilon, \delta}_{S_0}:=1$.
\item Set $j=1$.
\item Set $s_j=0$.
\item Set $S_j=S$.
\item Request ${p}^{\epsilon, \delta}_{S_j}$ from the {{}}{multiplicative precision}  simulator.
\item Compute $c_j:={p}^{\epsilon, \delta}_{S_j}/{p}^{\epsilon, \delta}_{S_{j-1}}$
\item Sample $\tilde{r}$ uniformly from $\bitstring{m}$.
\item Compute $r=\sum_{i=1}^{m} {\tilde{r}_i}2^{-i}$
\item If $c_j\geq r$, then set $\tilde{X}_j=0$ otherwise, set $\tilde{X}_j=1$.
\item Set $s_j=\tilde{X}_j$.
\item If $j=n$, output the string $\tilde{X}=(\tilde{X}_1,\ldots,\tilde{X}_n)$ and end.
\item Reset $j\rightarrow j+1$ and go to step 5.
\end{enumerate}

We note that multiplicative precision estimate can divide each other and still produce a multiplicative precision estimate. Hence $c_j$ computed in step 8 is a multiplicative precision estimate of the quantum conditional probability $p_{S_j}/p_{S_{j-1}}=\pr(X_j=x_j~|~X_1=x_1, \ldots, X_{j-1}=x_{j-1})$. This ensures that for $\epsilon'=\frac{1}{poly(n)}$, there exist polynomials $f_1, f_2, f_3$ such that $\epsilon\leq 1/f_1(n)$, $\delta\leq 1/f_2(n)$ and $m\geq f_3(n)$ satisfy the desired accuracy.

\section{On poly-sparsity and anti-concentration}\label{ps vs ac section}

In this section we will prove that poly-sparsity and anti-concentration cannot simultaneously be satisfied by any family of quantum circuits. This result is proven in Theorem \ref{ac vs ps thm}.

The condition of poly-sparsity forces the output distributions over exponentially many outcomes to concentrate on polynomially many outcomes. Alternatively, the property of anti-concentration forces the probabilities of observing any particular outcome, over random choices of circuits, to be low. Intuitively these properties do appear to oppose each other. However, since these properties are statements with respect to different probability spaces, we must first translate each property into a statement about a common probability space and with a common measure. This is done for anti-concentration and poly-sparsity in Lemma \ref{ac implication} and \ref{ps implication} respectively. We then state and prove our main claim in Theorem \ref{ac vs ps thm}.

First let us restate the relevant definitions is some detail.

\defn{(poly-sparse) Let $\mc{P}$ be a discrete probability distribution. We say that $\mc{P}$ is poly-sparse if there exists a polynomial $P(x)$ such that for all $\epsilon > 0$, $\mc{P}$ is \ets~whenever $t\geq P(\frac{1}{\epsilon})$. 

Let $\mathbb{P}$ be a family of probability distributions with $\mc{P}_{a}\in \mathbb{P}$ a distribution over $\bitstring{k_{a}}$.  We say that $\mathbb{P}$ is poly-sparse if there exists a polynomial $P(x)$ such that for all $\epsilon>0$ and $a \in \mc{A}^{*}$, $\mc{P}_{a}$ is \ets~whenever $t\geq P(k_{a}/\epsilon)$.}

\defn{(anti-concentration) Let $\mathcal{C}$ be a family of quantum circuits with $\mathbb{P}$ its associated family of probability distributions. For all $n\in \bbn$  let $\sigma_n$ be a probability measure over $\mathcal{A}^n$. We say that $\mathcal{C}$ anti-concentrates with respect to the set of measures $\Sigma:=\set{\sigma_n}_{n\in\bbn}$ iff $\forall n\in \bbn$, $\forall x\in \bitstring{n}$ and $\forall \alpha\in (0,1)$:
	\begin{align}
		\underset{a\in \mathcal{A}^n}{\pr} \left(\mathcal{P}_a (x) \geq \frac{\alpha}{2^n} \right)> \frac{(1-\alpha)^2}{2}\label{anticoncentration detailed}
	\end{align}
	where the probability is with respect to the measure $\sigma_n$.
	}
	
\begin{lemma}\label{ac implication}
For each $n\in \bbn$, let $\sigma_n$ be a probability measure over $\mathcal{A}^n$, let $\nu_n$ be any probability measure over over $\bitstring{n}$ and let $\tau_n$ be a probability measure over $\mathcal{A}^n \times \bitstring{n}$ defined as the product measure $\sigma_n \times \nu_n$. Then $\mathcal{C}$ anti-concentrates with respect to $\set{\sigma_n}_{n\in\bbn}$ implies that $\forall n\in \bbn$ and $\forall \alpha\in (0,1)$:
\begin{align}
	\underset{(a,x) \in \mathcal{A}^n\times \bitstring{n}}{\pr} \left(\mathcal{P}_a (x) \geq \frac{\alpha}{2^n} \right)> \frac{(1-\alpha)^2}{2}
\end{align}
where the probability is taken with respect to $\tau_n$.
\end{lemma}
\begin{proof}
For each $n\in \bbn$, we first define the sets $S_x:=\set{(a,x)\in \mathcal{A}^n\times \bitstring{n}~|~\mathcal{P}_a (x) \geq \frac{\alpha}{2^n}}$ and $S'_x:=\set{a\in \mathcal{A}^n~|~\mathcal{P}_a (x) \geq \frac{\alpha}{2^n}}$. Let $S:=\underset{x\in \bitstring{n}}{\cup} S_x$.

By the definition of anti-concentration, we have that $\forall n\in \bbn$, $\forall x\in \bitstring{n}$ and $\forall \alpha\in (0,1)$:
	\begin{align*}
		\sigma_n (S_x)> \frac{(1-\alpha)^2}{2}\text{.}
	\end{align*}
	Let us fix $\nu_n$ and note that:
	\begin{align}
\sum_{x\in \bitstring{n}} \nu_n (\set{x})\times \sigma_n (S'_x)&> \sum_{x\in \bitstring{n}} \nu_n (\set{x})\times\frac{(1-\alpha)^2}{2}\label{ac measure switch}
	\end{align}
	
	The LHS of Eq.~\eqref{ac measure switch} simplifies as follows:
	\begin{align*}
			\sum_{x\in \bitstring{n}} \nu_n (\set{x})\times \sigma_n (S'_x)&=\sum_{x\in \bitstring{n}} \tau_n (S_x)\\
			&=\tau_n(S)\text{.}
	\end{align*}
	
	By the fact that $\nu_n$ is a probability measure, the RHS simplifies to $\frac{(1-\alpha)^2}{2}$ thus proving the claim.
\end{proof}

\begin{lemma}\label{ps implication}
For each $n\in \bbn$, let $\sigma_n$ be any probability measure over $\mathcal{A}^n$, let $u_n$ be the uniform probability measure over over $\bitstring{n}$ and let $\tau_n$ be a probability measure over $\mathcal{A}^n \times \bitstring{n}$ defined as the product measure $\sigma_n \times \mu_n$. Then $\mathcal{C}$ is poly-sparse implies that $\forall \beta, \epsilon \in (0,1]$,   $\exists n_0 \in \bbn$ such that $\forall n \geq n_0$ and $\forall \gamma>0$:
\begin{align}
	\underset{(a,x) \in \mathcal{A}^n\times \bitstring{n}}{\pr} \left(\mathcal{P}_a (x) \geq \frac{\epsilon}{2^n \gamma (1-\beta)} \right)\leq \gamma+\beta
\end{align}
where the probability is taken with respect to $\tau_n$.

\end{lemma}
\begin{proof}
For any family of quantum circuits $\mathcal{C}$, $\forall n\in \bbn$ and for each $a\in \mathcal{A}^n$, a minimal function $t_a: [0,1]\rightarrow \bbn$ can be uniquely defined such that $\forall \epsilon \in [0,1]$, the probability of observing an outcome to be one of the $t_a(\epsilon)$ most likely outcomes (when circuit $c_a\in \mathcal{C}$ is run) is $\geq 1-\epsilon$. Poly-sparsity implies that there exists a polynomial $P$ such that $\forall \epsilon \in (0,1]$, $\forall n\in \bbn$, $\forall a\in \mathcal{A}^n$, $t_a(\epsilon)\leq P(n/\epsilon)$.

We apply Markov's inequality, which states that for $R$ a non-negative random variable and $\gamma>0$:
\begin{align*}
	\pr\left(R\geq \frac{\mathbb{E}[ R]}{\gamma} \right)\leq \gamma \text{.} \qquad \text{(Markov's inequality)}
\end{align*}

For any fixed $n\in \bbn$, $a\in \mathcal{A}^n$, let us consider uniformly randomly sampling from one of the $2^n-t_a(\epsilon)$ least likely outcomes. In this case:
\begin{align*}
	\mathbb{E}[\mathcal{P}_a (x)] &=\frac{\epsilon^+}{2^n-t_a(\epsilon)}\\
	&\leq \frac{\epsilon}{2^n-t_a(\epsilon^{+})}\\
	&= \frac{\epsilon}{2^n-t_a(\epsilon)-1}
\end{align*}
where $\epsilon^{\pm}$ are the two extremal points of the interval $V:=[\epsilon^-,\epsilon^+)$ containing $\epsilon$ such that $\forall \kappa \in V$ $t_a(\kappa)=t_a(\epsilon)$.

This implies that $\forall \epsilon \in (0,1]$, $\forall n\in \bbn$, $\forall a\in \mathcal{A}^n$, $\forall \gamma>0$, if we uniformly randomly sample from one of the $2^n-t_a(\epsilon)$ least likely outcomes, then:
\begin{align*}
	\underset{x \in \bitstring{n}}{\pr}\left(\mathcal{P}_a (x) \geq \frac{\mathbb{E}[\mathcal{P}_a (x)]}{\gamma} \right) &\leq\underset{x \in \bitstring{n}}{\pr}\left(\mathcal{P}_a (x) \geq \frac{\epsilon}{\gamma (2^n-t_a(\epsilon)-1)} \right)\\
	&\leq \gamma
\end{align*}

Hence, for $x$ uniformly sampled over \emph{all} bit-strings:
\begin{align*}
	\underset{x \in \bitstring{n}}{\pr}\left(\mathcal{P}_a (x) \geq \frac{\epsilon}{\gamma (2^n-t_a(\epsilon)-1)} \right) \leq \gamma+\frac{t_a(\epsilon)}{2^n}
\end{align*}
i.e.
\begin{align}
	\underset{x \in \bitstring{n}}{\pr}\left(\mathcal{P}_a (x) \geq \eta \right) \leq c \label{poly-sparsity markov}
\end{align}
where $\eta:=\frac{\epsilon 2^{-n}}{\gamma (1-t_a(\epsilon)2^{-n}-2^{-n})}$, $c:=\gamma+\frac{t_a(\epsilon)}{2^n}$ and the probability is over the uniform measure $u_n$ over $\bitstring{n}$.

Let us note that by the fact that poly-sparsity requires that there is a polynomial $P$ such that $t_a(\epsilon)\leq P(n/\epsilon)$, we have that $\forall \beta, \epsilon \in (0,1]$,   $\exists n_0 \in \bbn$ such that $\forall n \geq n_0$, $\forall a\in \mathcal{A}^n$ and $\forall \gamma>0$:
\begin{align*}
	\beta \geq \frac{t_a(\epsilon)+1}{2^n}
\end{align*}
implying that $\eta \leq \frac{\epsilon 2^{-n}}{\gamma (1-\beta)}$ and $c<\gamma+\beta$.

For all $n\in \bbn$, we define the sets 
\begin{displaymath}
  T_a:=\set{(a,x)\in \mathcal{A}^n \times \bitstring{n}~|~\mathcal{P}_a (x) \geq \eta},
\end{displaymath} 
$T'_a:=\set{x\in  \bitstring{n}~|~\mathcal{P}_a (x) \geq \eta}$ and $T:=\underset{a\in \mathcal{A}^n}{\cup} T_a$. We now note that Eq.~\eqref{poly-sparsity markov} can be rewritten as:
\begin{align*}
	u_n(T'_a) \leq c \text{.}
\end{align*}
For any fixed sequence of measures $\sigma_n$ over $\mathcal{A}^n$, we can define $\tau_n=\sigma_n \times u_n$ to get:
\begin{align}
	\sum_{a\in \mathcal{A}^n} \sigma_n({a}) \times u_n(T'_a) \leq \sum_{a\in \mathcal{A}^n} \sigma_n({a}) \times c \label{ps measure switch} \text{.}
\end{align}

The LHS of this equation can be simplified as follows:
\begin{align*}
	\sum_{a\in \mathcal{A}^n} \sigma_n({a}) \times u_n(T'_a)&=\sum_{a\in \mathcal{A}^n} \tau_n(T_a)\\
	&=\tau_n(T)\text{.}
\end{align*}

By the fact that $\sigma_n$ is a probability measure, the RHS of Eq.~\eqref{ps measure switch} simplifies to $c$ thus proving the claim.
\end{proof}

\begin{thm}\label{ac vs ps thm}
For each $n\in \bbn$ let $\sigma_n$ be a measures over $\mathcal{A}^n$ and let $\mathcal{C}=\underset{n\in \bbn}{\cup} \set{c_a}_{a\in \mathcal{A}^n}$ be a family of quantum circuits such that $\mathcal{C}$ anti-concentrates with respect to $\Sigma=\set{\sigma_n}_{n\in \bbn}$. Then $\mathcal{C}$ is not poly-sparse.
\end{thm}
\begin{proof}
We will show that the two conditions together give rise to a contradiction and hence are inconsistent. We apply Lemma \ref{ac implication} and set $\alpha=1/8$ and for each $n\in \bbn$, $\nu_n=u_n$, the uniform measure over $\bitstring{n}$ giving:

\begin{align}
	\underset{(a,x) \in \mathcal{A}^n\times \bitstring{n}}{\pr} \left(\mathcal{P}_a (x) \geq \frac{1}{2^n \times 2} \right)> \frac{1}{8}\text{.}\label{ac cond in thm}
\end{align}

We apply Lemma \ref{ps implication} and set $\beta=\gamma=1/16$ and $\epsilon=\frac{1}{64}\left(1-\frac{1}{16} \right)$ giving:

\begin{align}
	\underset{(a,x) \in \mathcal{A}^n\times \bitstring{n}}{\pr} \left(\mathcal{P}_a (x) \geq \frac{1}{2^n \times 4} \right)\leq \frac{1}{8} \label{ps cond in thm}
\end{align}
where both Eq.~\eqref{ac cond in thm} and \eqref{ps cond in thm} are with respect to the same measure $\tau_n$ thus implying a contradiction.

\end{proof}

Theorem \ref{ac vs ps thm} establishes that poly-sparsity and anti-concentration are mutually exclusive properties. However, these properties can nonetheless jointly fail to be satisfied. For example, we can take an infinite family of circuits (growing faster than $2^n$) which is poly-sparse then for each $n$ append a single circuit with output distribution that is uniform. This change ensures that the family now breaks poly-sparsity but is insufficient for anti-concentration to be instated. In addition, if we assume that the family admits a poly-box then we notice that by Theorem \ref{simable} this family was \esimable~before the change but after the change no longer satisfies the requirements for the application of Theorem \ref{simable} despite the fact that we have only added a sequence of uniform distributions which are \esimable. A clean mathematical characterization of when poly-sparsity and anti-concentration can jointly fail may help inform how to best resolve this undesirable situation.


\begin{thebibliography}{99}

\bibitem{feynman1982simulating}
R.~P. Feynman, ``{Simulating Physics with Computers},'' \href{https://doi.org/10.1007/BF02650179}{Int. J. Theor.
  Phys., \textbf{21}, 467--488 (1982).}

\bibitem{aaronson2004improved}
S.~Aaronson and D.~Gottesman, ``Improved simulation of stabilizer circuits,'' \href{https://doi.org/10.1103/PhysRevA.70.052328}{Phys. Rev. A, \textbf{70}, 052328 (2004).}

\bibitem{valiant2002quantum}
L.~G. Valiant, ``Quantum circuits that can be simulated classically in
  polynomial time,'' \href{https://doi.org/10.1137/S0097539700377025}{SIAM Journal on Computing, \textbf{31}, 1229--1254 (2002).}

\bibitem{terhal2002classical}
B.~M. Terhal and D.~P. DiVincenzo, ``Classical simulation of
  noninteracting-fermion quantum circuits,'' \href{https://doi.org/10.1103/PhysRevA.65.032325}{ Phys. Rev. A \textbf{65}, 032325 (2002).}

\bibitem{bartlett2002efficient}
S.~D. Bartlett, B.~C. Sanders, S.~L. Braunstein, and K.~Nemoto, ``Efficient
  classical simulation of continuous variable quantum information processes,''
  \href{https://doi.org/10.1103/PhysRevLett.88.097904}{Phys. Rev. Lett. \textbf{88}, 097904 (2002).}

\bibitem{gurvits2005complexity}
L.~Gurvits, ``On the complexity of mixed discriminants and related problems,'' \href{https://doi.org/10.1007/11549345_39}{in: Jedrzejowicz J., Szepietowski A. (eds) Mathematical Foundations of Computer Science 2005. MFCS 2005. Lecture Notes in Computer Science, vol 3618. Springer, Berlin, Heidelberg (2005).}

\bibitem{aaronson2011computational}
S.~Aaronson and A.~Arkhipov, ``The computational complexity of linear optics,'' \href{https://doi.org/10.1145/1993636.1993682}{
  in {\em Proceedings of the forty-third annual ACM symposium on Theory of
  computing}, 333--342, ACM (2011).}

\bibitem{bremner2010classical}
M.~J. Bremner, R.~Jozsa, and D.~J. Shepherd, ``Classical simulation of
  commuting quantum computations implies collapse of the polynomial
  hierarchy,'' \href{https://doi.org/10.1098/rspa.2010.0301}{Proc. R. Soc. A \textbf{467}, (2010).}

\bibitem{bremner2016average}
M.~J. Bremner, A.~Montanaro, and D.~J. Shepherd, ``Average-case complexity
  versus approximate simulation of commuting quantum computations,'' \href{https://doi.org/10.1103/PhysRevLett.117.080501}{Phys. Rev. Lett. \textbf{117}, 080501 (2016).}

\bibitem{gao2017quantum}
X.~Gao, S.-T. Wang, and L.-M. Duan, ``Quantum supremacy for simulating a
  translation-invariant Ising spin model,'' \href{https://doi.org/10.1103/PhysRevLett.118.040502}{Phys. Rev. Lett. \textbf{118}, 040502 (2017).}

\bibitem{bermejo2017architectures}
J.~Bermejo-Vega, D.~Hangleiter, M.~Schwarz, R.~Raussendorf, and J.~Eisert,
  ``Architectures for quantum simulation showing a quantum speedup,'' \href{https://doi.org/10.1103/PhysRevX.8.021010}{Phys. Rev. X \textbf{8}, 021010 (2018).}

\bibitem{fefferman2015power}
B.~Fefferman and C.~Umans, ``The power of quantum fourier sampling,'' {\em
  arXiv preprint} \href{http://arXiv.org/abs/1507.05592}{arXiv:1507.05592, (2015)}.

\bibitem{morimae2014hardness}
T.~Morimae, K.~Fujii, and J.~F. Fitzsimons, ``Hardness of classically
  simulating the one-clean-qubit model,'' \href{https://doi.org/10.1103/PhysRevLett.112.130502}{Phys. Rev. Lett. \textbf{112}, 130502 (2014).}

\bibitem{morimae2017power}
T.~Morimae, K.~Fujii, and H.~Nishimura, ``Power of one nonclean qubit,'' \href{https://doi.org/10.1103/PhysRevA.95.042336}{Phys. Rev. A \textbf{95}, 042336 (2017).}

\bibitem{boixo2018characterizing}
S.~Boixo, S.~V. Isakov, V.~N. Smelyanskiy, R.~Babbush, N.~Ding, Z.~Jiang, M.~J.
  Bremner, J.~M. Martinis, and H.~Neven, ``Characterizing quantum supremacy in
  near-term devices,'' \href{https://doi.org/10.1038/s41567-018-0124-x}{Nature Phys. \textbf{14}, 595-600 (2018).}

\bibitem{bouland2017quantum}
A.~Bouland, J.~F. Fitzsimons, and D.~E. Koh, ``{Complexity Classification of
  Conjugated Clifford Circuits},'' \href{https://doi.org/10.4230/LIPIcs.CCC.2018.21}{in 33rd Computational Complexity
  Conference (CCC 2018) (R.~A. Servedio, ed.), vol.~102 of {\em Leibniz
  International Proceedings in Informatics (LIPIcs)}, (Dagstuhl, Germany),
  pp.~21:1--21:25, Schloss Dagstuhl--Leibniz-Zentrum fuer Informatik, 2018.}

\bibitem{brod2016efficient}
D.~J. Brod, ``Efficient classical simulation of matchgate circuits with
  generalized inputs and measurements,'' \href{https://doi.org/10.1103/PhysRevA.93.062332}{Phys. Rev. A \textbf{93}, 062332 (2016).}

\bibitem{Jozsa2008matchgates}
R.~Jozsa and A.~Miyake, ``Matchgates and classical simulation of quantum
  circuits,'' \href{https://doi.org/10.1098/rspa.2008.0189}{Proc. R. Soc. A \textbf{464}, 3089--3106 (2008).}

\bibitem{veitch2012negative}
V.~Veitch, C.~Ferrie, D.~Gross, and J.~Emerson, ``Negative quasi-probability as
  a resource for quantum computation,'' \href{https://doi.org/10.1088/1367-2630/14/11/113011}{New J. Phys. \textbf{14}, 113011 (2012).}

\bibitem{veitch2013CV}
V.~Veitch, N.~Wiebe, C.~Ferrie, and J.~Emerson, ``Efficient simulation scheme
  for a class of quantum optics experiments with non-negative Wigner
  representation,'' \href{https://doi.org/10.1088/1367-2630/15/1/013037}{New J. Phys. \textbf{15}, 013037 (2013).}

\bibitem{mari2012positive}
A.~Mari and J.~Eisert, ``Positive Wigner functions render classical
  simulation of quantum computation efficient,'' \href{https://doi.org/10.1103/PhysRevLett.109.230503}{Phys. Rev. Lett. \textbf{109}, 230503 (2012).}

\bibitem{bravyi2016improved}
S.~Bravyi and D.~Gosset, ``Improved classical simulation of quantum circuits
  dominated by {C}lifford gates,'' \href{https://doi.org/10.1103/PhysRevLett.116.250501}{Phys. Rev. Lett. \textbf{116}, 250501 (2016).}

\bibitem{bennink2017unbiased}
R.~S. Bennink, E.~M. Ferragut, T.~S. Humble, J.~A. Laska, J.~J. Nutaro, M.~G.
  Pleszkoch, and R.~C. Pooser, ``Unbiased simulation of near-{C}lifford quantum
  circuits,'' \href{https://doi.org/10.1103/PhysRevA.95.062337}{Phys. Rev. A \textbf{95}, 062337 (2017).}

\bibitem{bremner2017achieving}
M.~J. Bremner, A.~Montanaro, and D.~J. Shepherd, ``Achieving quantum supremacy
  with sparse and noisy commuting quantum computations,'' \href{https://doi.org/10.22331/q-2017-04-25-8}{Quantum \textbf{1}, 8 (2017).}

\bibitem{oszmaniec2018classical}
M.~Oszmaniec and D.~J. Brod, ``Classical simulation of photonic linear optics
  with lost particles,'' \href{https://doi.org/10.1088/1367-2630/aadfa8}{New J. Phys. \textbf{20}, 092002 (2018).}

\bibitem{howard2017application}
M.~Howard and E.~Campbell, ``Application of a resource theory for magic states
  to fault-tolerant quantum computing,'' \href{https://doi.org/10.1103/PhysRevLett.118.090501}{Phys. Rev. Lett. \textbf{118}, 090501 (2017).}

\bibitem{pashayan2015estimating}
H.~Pashayan, J.~J. Wallman, and S.~D. Bartlett, ``Estimating outcome
  probabilities of quantum circuits using quasiprobabilities,'' \href{https://doi.org/10.1103/PhysRevLett.115.070501}{Phys. Rev. Lett. \textbf{115}, 070501 (2015).}

\bibitem{nest2012efficient}
M.~Van Den Nest, ``Efficient classical simulations of quantum {F}ourier transforms and
  normalizer circuits over {A}belian groups,'' {\em
  arXiv preprint} \href{http://arXiv.org/abs/1201.4867}{arXiv:1201.4867, 2012.}

\bibitem{bermejo2014classical}
J.~Bermejo-Vega and M.~Van Den~Nest, ``Classical simulations of {A}belian-group
  normalizer circuits with intermediate measurements,'' Quantum Info. Comput. \textbf{14}, 181--216 (2014).

\bibitem{schwarz2013simulating}
M.~Schwarz and M.~Van Den Nest, ``Simulating quantum circuits with sparse output
  distributions,'' {\em
  arXiv preprint} \href{http://arXiv.org/abs/1310.6749}{arXiv:1310.6749, 2013.}

\bibitem{bravyi2016trading}
S.~Bravyi, G.~Smith, and J.~A. Smolin, ``Trading classical and quantum
  computational resources,'' \href{https://doi.org/10.1103/PhysRevX.6.021043}{Phys. Rev. X, \textbf{6}, 021043 (2016).}

\bibitem{temme2017error}
K.~Temme, S.~Bravyi, and J.~M. Gambetta, ``Error mitigation for short-depth
  quantum circuits,'' \href{https://doi.org/10.1103/PhysRevLett.119.180509}{Phys. Rev. Lett. \textbf{119}, 180509 (2017).}

\bibitem{nest2008classical}
M.~Van Den~Nest, ``Classical simulation of quantum computation, the
  {G}ottesman-{K}nill theorem, and slightly beyond,'' Quantum Info. Comput. \textbf{10}, 258--271 (2010).

\bibitem{jozsa2003role}
R.~Jozsa and N.~Linden, ``On the role of entanglement in quantum-computational
  speed-up,'' \href{https://doi.org/10.1098/rspa.2002.1097}{Proc. R. Soc. A \textbf{459}, 2011--2032 (2003).}

\bibitem{Terhal2004}
B.~M. Terhal and D.~P. DiVincenzo, ``Adaptive quantum computation, constant
  depth quantum circuits and arthur-merlin games,'' Quantum Info. Comput. \textbf{4}, 134--145 (2004).

\bibitem{morimae2017hardness}
T.~Morimae, ``Hardness of classically sampling the one-clean-qubit model with
  constant total variation distance error,'' \href{https://doi.org/10.1103/PhysRevA.96.040302}{Phys. Rev. A \textbf{96}, 040302 (2017).}

\bibitem{aaronson2014equivalence}
S.~Aaronson, ``The equivalence of sampling and searching,'' \href{https://doi.org/10.1007/s00224-013-9527-3}{Theory of
  Computing Systems, \textbf{55}, 281--298 (2014).}

\bibitem{nest2009simulating}
M.~Van Den~Nest, ``Simulating quantum computers with probabilistic methods,''
  Quantum Info. Comput. \textbf{11}, 784--812 (2011).

\bibitem{shepherd2010binary}
D.~Shepherd, ``Binary matroids and quantum probability distributions,'' {\em
  arXiv preprint} \href{http://arXiv.org/abs/1005.1744}{arXiv:1005.1744, (2010).}

\bibitem{Bravyi2005magic}
S.~Bravyi and A.~Kitaev, ``Universal quantum computation with ideal clifford
  gates and noisy ancillas,'' \href{https://doi.org/10.1103/PhysRevA.71.022316}{Phys. Rev. A \textbf{71}, 022316 (2005).}

\bibitem{jozsa2014classical}
R.~Jozsa and M.~Van Den~Nest, ``Classical simulation complexity of extended
  {C}lifford circuits,'' Quantum Info. Comput. \textbf{14}, 633--648 (2014).

\bibitem{goldberg2017complexity}
L.~A. Goldberg and H.~Guo, ``The complexity of approximating complex-valued
  {I}sing and {T}utte partition functions,'' \href{https://doi.org/10.1007/s00037-017-0162-2}{comput. complex. \textbf{26}, 765--833 (2017).}

\bibitem{kuperberg2015hard}
G.~Kuperberg, ``How hard is it to approximate the {J}ones polynomial?,'' \href{https://doi.org/10.4086/toc.2015.v011a006}{Theory of Computing, \textbf{11}, 183--219 (2015).}

\bibitem{Fujii2017commuting}
K.~Fujii and T.~Morimae, ``Commuting quantum circuits and complexity of Ising
  partition functions,'' \href{https://doi.org/10.1088/1367-2630/aa5fdb}{New J. Phys. \textbf{19}, 033003 (2017).}

\bibitem{stahlke2014quantum}
D.~Stahlke, ``Quantum interference as a resource for quantum speedup,'' \href{https://doi.org/10.1103/PhysRevA.90.022302}{Phys. Rev. A \textbf{90}, 022302 (2014).}

\bibitem{shepherd2009temporally}
D.~Shepherd and M.~J. Bremner, ``Temporally unstructured quantum computation,''
  \href{https://doi.org/10.1098/rspa.2008.0443}{Proc. R. Soc. A \textbf{465}, 1413--1439 (2009).}

\bibitem{shepherd2010quantum}
D.~J. Shepherd, ``Quantum complexity: restrictions on algorithms and
  architectures,'' {\em
  arXiv preprint} \href{http://arXiv.org/abs/1005.1425}{arXiv:1005.1425 (2010).}

\bibitem{vertigan1998bicycle}
D.~Vertigan, ``Bicycle dimension and special points of the {T}utte
  polynomial,'' \href{https://doi.org/10.1006/jctb.1998.1860}{J. Combin. Theory Ser. B \textbf{74}, 378--396 (1998).}

\bibitem{Renner2004Smooth}
R.~{Renner} and S.~{Wolf}, ``Smooth R\'{e}nyi entropy and applications,'' \href{https://doi.org/10.1109/ISIT.2004.1365269}{in International Symposium on Information Theory, 2004. ISIT 2004. Proceedings.,  pp.~233 (2004).}

\bibitem{stockmeyer1983complexity}
L.~Stockmeyer, ``The complexity of approximate counting,'' \href{https://doi.org/10.1145/800061.808740}{in {\em Proceedings
  of the fifteenth annual ACM symposium on Theory of computing}, pp.~118--126,
  ACM, (1983).}

\bibitem{toda1991pp}
S.~Toda, ``{PP} is as hard as the polynomial-time hierarchy,'' \href{https://doi.org/10.1137/0220053}{SIAM
  Journal on Computing \textbf{20}, 865--877 (1991).}

\bibitem{aaronson2014generalizing}
S.~Aaronson and T.~Hance, ``Generalizing and derandomizing {G}urvits's
  approximation algorithm for the permanent,'' Quantum Info. Comput. \textbf{14}, 541--559 (2014).

\end{thebibliography}
\end{document}